\documentclass[11pt]{article}

\usepackage{preamble}

\usepackage{fix-cm}

\usepackage{cite}
\usepackage{orcidlink}
\usepackage{bbm}
\usepackage{makecell}
\usepackage{mathtools}

\title{Automated Reencoding Meets Graph Theory}

\author{Benjamin Przybocki\,\orcidlink{0009-0007-5489-1733}}
\author{Bernardo Subercaseaux\,\orcidlink{0000-0003-2295-1299}}
\author{Marijn J. H. Heule\,\orcidlink{0000-0002-5587-8801}}
\affil{Carnegie Mellon University \\  {\upshape\ttfamily [bprzyboc,bsuberca,mheule]@andrew.cmu.edu}}
\date{}

\newcommand{\edgegreen}{green!65!black}

\begin{document}

\maketitle 
\begin{abstract}
\emph{Bounded Variable Addition} (BVA) is a central preprocessing method in modern state-of-the-art SAT solvers. We provide a graph-theoretic characterization of which 2-CNF encodings can be constructed by an idealized BVA algorithm. Based on this insight, we prove new results about the behavior and limitations of BVA and its interaction with other preprocessing techniques. We show that idealized BVA, plus some minor additional preprocessing (e.g., equivalent literal substitution), can reencode any 2-CNF formula with $n$ variables into an equivalent 2-CNF formula with $(\tfrac{\lg(3)}{4}+o(1))\,\tfrac{n^2}{\lg n}$ clauses. Furthermore, we show that without the additional preprocessing the constant factor worsens from $\tfrac{\lg(3)}{4} \approx 0.396$ to $1$, and that no reencoding method can achieve a constant below $0.25$. On the other hand, for the at-most-one constraint on $n$ variables, we prove that idealized BVA cannot reencode this constraint using fewer than $3n-6$ clauses, a bound that we prove is achieved by actual implementations. In particular, this shows that the product encoding for at-most-one, which uses $2n+o(n)$ clauses, cannot be constructed by BVA regardless of the heuristics used. Finally, our graph-theoretic characterization of BVA allows us to leverage recent work in algorithmic graph theory to develop a drastically more efficient implementation of BVA that achieves a comparable clause reduction on random monotone 2-CNF formulas.
\end{abstract}

\section{Introduction}

Providing a satisfactory explanation for the remarkable performance of SAT solvers in practice is a major theoretical challenge. Part of the difficulty of explaining this phenomenon, further detailed in the appropriately titled article ``On the Unreasonable Effectiveness of SAT Solvers'' by Vardi and Ganesh~\cite{Ganesh_Vardi_2021}, is that modern SAT solvers combine a variety of heuristics, preprocessing and inprocessing techniques, data structures, and so on. 
We posit that, to make progress on the overall question, we ought to better understand the power and limitations of each of these ingredients. In this paper, we home in on one of these components, namely, \emph{Bounded Variable Addition} (BVA)~\cite{mantheyAutomatedReencodingBoolean2012}, a preprocessing technique that introduces auxiliary variables to reencode an input formula into an equisatisfiable formula with fewer clauses. 
In 2023, a successor of BVA named \emph{Structured BVA} (SBVA)~\cite{haberlandtEffectiveAuxiliaryVariables2023}, led the solver~\textsf{SBVA-CaDiCaL}, whose primary innovation was to incorporate SBVA as preprocessing to the state-of-the-art solver~\textsf{CaDiCaL}, to win the 1st place in the 2023 SAT Competition~\cite{satComp23}. At the 2024 SAT Competition, many other solvers followed suit by including some variant of BVA as a preprocessing step, including the competition winner \textsf{kissat-sc2024}~\cite{SATCompetition2024}. BVA is now incorporated into the stable versions of \textsf{CaDiCaL} and \textsf{Kissat} under the name \textsf{factor}~\cite{githubFactor}.

However, despite its empirically demonstrated utility, little is known about BVA and its variants from a theoretical point of view.
For example, Biere, one of the authors of the original BVA paper, commented~\cite{githubExactlyClauses}:
\begin{quote}
\emph{In~\cite{mantheyAutomatedReencodingBoolean2012}, we already realized that BVA reencodes [the direct encoding of the at-most-one constraint] into something better. Results were only empirical though and
we did not analyze whether in principle BVA can turn the direct encoding into for instance
the product or sequential counter encoding.}
\end{quote}
\noindent
In this paper, we develop a graph-theoretic framework for analyzing BVA, by means of which we can answer these and similar questions about the power and limitations of BVA. Furthermore, our framework allows us to draw on recent work in algorithmic graph theory to develop a drastically more efficient implementation of BVA.

\vspace{-0.8em}
\subsection*{Our Results}

We focus on the behavior of BVA on 2-CNF subformulas. Many formulas of interest, including, e.g., most SAT competition formulas, contain a large 2-CNF subformula that can be reencoded more compactly, leading to faster runtimes.
In~\Cref{sec:discussion}, we justify in more detail why the 2-CNF fragment not only captures the most important aspects of the theoretical analysis but also explains most of BVA's empirical success.

Our primary result is a characterization of which reencodings of a 2-CNF formula are achievable by BVA. Our characterization is in terms of \emph{rectifier networks}~\cite{lupanov, chistikov_et_al:LIPIcs.STACS.2017.23, ivanBicliqueCoveringsRectifier2014}, a graph-theoretic notion that has been studied in the context of circuit complexity. Specifically, we show in~\Cref{thm:characterization} that there is a correspondence between reencodings achievable by BVA and what we call \emph{strict polarized rectifier networks}. Our characterization holds for an \emph{idealized} version of BVA that abstracts away implementation-specific heuristics.

With this characterization in hand, we first study the class of 2-CNF formulas in general. By generalizing an old result of Nechiporuk~\cite{nechiporuk} about rectifier networks, we prove the following:\footnote{All logarithms in this paper are base 2 unless otherwise specified.}
\begin{theorem}[Informal] \label{thm-idealized-bva-no-simp}
    Idealized BVA can reencode every 2-CNF formula with $n$ variables to one with at most $(1+o(1)) \frac{n^2}{\lg n}$ clauses.
\end{theorem}
\noindent
Note that, without reencoding, 2-CNF formulas can have $\Theta(n^2)$ clauses. We also show that the bound from \Cref{thm-idealized-bva-no-simp} is sharp, in the sense that there is some 2-CNF formula that idealized BVA reencodes with at least $(1-o(1)) \frac{n^2}{\lg n}$ clauses. On the other hand, we describe a simplification routine that one can run before BVA (``pre-preprocessing''), which consists of equivalent literal substitution and a weaker form of failed literal elimination~\cite{heuleEfficientCNFSimplification2011}, and this simplification step improves the worst-case performance of idealized BVA:
\begin{restatable}[Informal]{theorem}{thm2cnf} \label{thm-2cnf}
    Idealized BVA with simplification can reencode every 2-CNF formula with $n$ variables to one with at most $(\frac{\lg(3)}{4}+o(1)) \frac{n^2}{\lg n}$ clauses.
\end{restatable}
\noindent
This bound is also sharp. Furthermore, we show that for every reencoding method whose output is 2-CNF, there is some 2-CNF formula that it reencodes with at least $(\frac{1}{4}-o(1)) \frac{n^2}{\lg n}$ clauses.
Thus, with respect to worst-case analysis, idealized BVA is optimal up to a constant factor.

An important subset of 2-CNF formulas is the monotone (or antitone) 2-CNF formulas. For example, a direct encoding for finding an independent set in a graph $G$ would include, for every edge $\{u, v\} \in E(G)$, a clause
\(
(\overline{x_u} \lor \overline{x_v}).
\)
We can actually improve the bound from \Cref{thm-idealized-bva-no-simp} for monotone 2-CNF formulas:
\begin{restatable}[Informal]{theorem}{thm2cnfmonotone} \label{thm-2cnf-monotone}
    Idealized BVA can reencode every monotone 2-CNF formula with $n$ variables to one with at most $(\frac{1}{4}+o(1)) \frac{n^2}{\lg n}$ clauses.
\end{restatable}
\noindent
This bound is sharp, not merely for idealized BVA but for arbitrary reencoding methods whose output is 2-CNF.

Next, we study BVA's performance on a particular class of 2-CNF formulas, namely the direct encoding of the at-most-one constraint on $n$ variables:
\[
    \textsf{AtMostOne}(x_1,\dots,x_n) \coloneqq \bigwedge_{1 \le i < j \le n} (\overline{x_i} \lor \overline{x_j}).
\]
Manthey, Heule, and Biere~\cite{mantheyAutomatedReencodingBoolean2012} observed experimentally that their implementation of BVA reencodes $\textsf{AtMostOne}(x_1,\dots,x_n)$ into $3n-6$ clauses. We prove in \Cref{thm-amo-practice} that, given their heuristics, BVA always achieves this bound. On the other hand, via a combinatorial analysis of strict rectifier networks, we show a matching lower bound, regardless of the heuristics used:
\begin{restatable}[Informal]{theorem}{thmamo} \label{thm:amo}
    Idealized BVA cannot encode $\textsf{AtMostOne}(x_1,\dots,x_n)$ using fewer than $3n-6$ clauses.
\end{restatable}
\noindent In particular, this implies that the product encoding~\cite{chen2010new} for $\textsf{AtMostOne}(x_1,\dots,x_n)$, which uses $2n+o(n)$ clauses, cannot be constructed by idealized BVA.

\subparagraph*{Implementation.} Leveraging recent work by Krapivin et al.~\cite{stoc} on efficient algorithms for constructing biclique partitions of graphs, which are a special case of rectifier networks, we present an implementation of BVA for 2-CNF formulas that
runs in time $O(n^2)$. While the time complexity of the algorithm used in \textsf{factor}, the previous state of the art, has not been rigorously analyzed, it is at least $\Omega(n^3)$. This algorithmic improvement leads to an order-of-magnitude speedup over \textsf{factor} on large random monotone 2-CNF formulas.

\subparagraph*{Proofs.} For improved readability, most proofs are postponed to the appendix.

\section{Preliminaries}\label{sec:preliminaries}
\subsection{Notation and Terminology}

We will use $\top$ and $\bot$ to denote \emph{true} and \emph{false}, respectively. The negation of a variable $x$ is denoted $\overline{x}$, and a \emph{literal} is either a variable or the negation of a variable. Literals $x$ and $\overline{x}$ are said to be complementary, for any variable $x$. A \emph{clause} is a set of non-complementary literals, and a \emph{formula} is a set of clauses. The \emph{size} of a formula is its number of clauses. We say a formula is \emph{2-CNF} if every clause has size exactly 2.
We denote the set of variables appearing in a formula $F$ as $\textsf{Var}(F)$, and we define $\textsf{Var}(C)$ analogously for a clause $C$. 
Given a set $\mathcal{V}$ of variables, an \emph{assignment} is a function $\tau\colon \mathcal{V} \to \{\bot, \top\}$.
For an assignment $\tau: \mathcal{V} \to \{\bot, \top\}$ and a formula $F$, now with $\mathcal{V} \subseteq \textsf{Var}(F)$, we denote by $F|_\tau$ the formula obtained by eliminating from $F$ each clause satisfied by $\tau$, and then from each remaining clause eliminating every literal $\ell$ such that $\tau \models \overline{\ell}$. Note that $\textsf{Var}(F|_\tau) = \textsf{Var}(F) \setminus \mathcal{V}$.
We will write $\textsf{SAT}(F)$ to say that $\tau \models F$ for some assignment $\tau$.

\subsection{CNF Encodings and BVA}



Boolean satisfiability problems, in general, are based on determining whether a boolean function $f \colon \{\bot,\top\}^n \to \{\bot, \top\}$ outputs $\top$ for some input. SAT solvers, however, take CNF formulas as input, and thus require an \emph{encoding} (formally defined next) of $f$ as a CNF formula.

\begin{definition}
    Given a boolean function $f \colon \{\bot,\top\}^n \to \{\bot, \top\}$, and a sequence of propositional variables $X := (x_1, \ldots, x_n)$, we say that a CNF formula $\varphi$ over variables $X \sqcup Y$ \emph{encodes} $f$ if, for every assignment $\tau$ of $X$,
    \[
    f(\tau(x_1), \dots, \tau(x_n)) = \top \iff  \textsf{SAT}(\varphi|_\tau).
    \]
    The variables in $Y$ are called \emph{auxiliary} variables, whereas those in $X$ are called \emph{base} variables.
\end{definition}

For example, the cardinality constraint \emph{``exactly $k$ variables from $x_1, \dots, x_n$ must be true''} is directly well defined as a boolean function, but there is significant freedom regarding how to encode it in CNF. An important component of this freedom lies in the auxiliary variables, since neither their number nor purpose is predetermined. 

Designing effective CNF encodings, and in particular choosing the right auxiliary variables to introduce, has been key to successfully tackling hard instances~\cite{heuleHappyEndingEmpty2024, subercaseauxPackingChromaticNumber2023d, boxFolding, kirchweger2025finiteinfinitesharperasymptotic}. However, the design of effective encodings can require significant technical expertise and domain-specific knowledge.
This makes \emph{automated reencoding} techniques, such as BVA, particularly appealing. Such techniques work by optimizing some quantitative proxy for the effectiveness of an encoding; a common choice of proxy is to consider some function of the number of clauses and the number of variables (e.g., their sum~\cite{mantheyAutomatedReencodingBoolean2012} or weighted sum~\cite{abioParametricApproachSmaller2013}), although more complicated metrics have also been proposed~\cite{haberlandtEffectiveAuxiliaryVariables2023}. Automated reencoding techniques operate by identifying patterned subformulas of a given CNF formula that can be rewritten using fewer clauses by introducing auxiliary variables. BVA, in particular, works exclusively by detecting and rewriting a single concrete pattern, which is illustrated by the following example.

\begin{example}\label{example:BVA-1}
    Consider a formula $F$ that includes the $9$ clauses of the form $(x_i \to x_j)$ for $i \in \{1,2,3\}$ and $j \in \{4,5,6\}$. These 9 clauses can be expressed more compactly by introducing an auxiliary variable $y$ and then enforcing:
    \[
        (x_1 \to y) \land (x_2 \to y) \land (x_3 \to y) \land (y \to x_4) \land (y \to x_5) \land (y \to x_6).
    \]
    This reduces the number of clauses from $9$ to $6$, at the cost of adding $1$ auxiliary variable. \Cref{fig:example-biclique} illustrates that, under the graph-theoretic perspective of this paper, the reencoded pattern is a \emph{biclique} (complete bipartite subgraph).
\end{example}
\begin{figure}[h]
    \begin{subfigure}{0.49\linewidth}
        \centering
        \begin{tikzpicture}[scale=0.65, basenodeleft/.style={circle, fill=blue!60!gray, text=white, inner sep=1pt}, basenoderight/.style={circle, fill=red!70!gray, text=white, inner sep=1pt}]
\node[draw, basenodeleft] (v1) at (0,0) {$x_1$};
\node[draw, basenodeleft] (v2) at (0,-1.5) {$x_2$};
\node[draw, basenodeleft] (v3) at (0,-3.0) {$x_3$};

\node[draw, basenoderight] (v4) at (3,0) {$x_4$};
\node[draw, basenoderight] (v5) at (3,-1.5) {$x_5$};
\node[draw, basenoderight] (v6) at (3,-3.0) {$x_6$};

\foreach \i in {1,2,3}{
  \foreach \j in {4,5,6} {
    \draw[thick, -Latex] (v\i) -- (v\j);
  }
}
\end{tikzpicture}
        \caption{Original encoding}
    \end{subfigure}
    \begin{subfigure}{0.49\linewidth}
        \centering
        \begin{tikzpicture}[scale=0.65, basenodeleft/.style={circle, fill=blue!60!gray, text=white, inner sep=1pt}, basenoderight/.style={circle, fill=red!70!gray, text=white, inner sep=1pt}, auxnode/.style={circle, fill=green!70!black, text=white, inner sep=2pt},]
\node[draw, basenodeleft] (v1) at (0,0) {$x_1$};
\node[draw, basenodeleft] (v2) at (0,-1.5) {$x_2$};
\node[draw, basenodeleft] (v3) at (0,-3.0) {$x_3$};
\node[draw, basenoderight] (v4) at (3, 0) {$x_4$};
\node[draw, basenoderight] (v5) at (3, -1.5) {$x_5$};
\node[draw, basenoderight] (v6) at (3, -3.0) {$x_6$};
\node[draw, auxnode] (aux) at (1.5, -1.5) {$y$};

\foreach \i in {1,2,3}{
    \draw[thick, -Latex] (v\i) -- (aux);
}
\foreach \j in {4,5,6}{
    \draw[thick, -Latex] (aux) -- (v\j);
}

\end{tikzpicture}
        \caption{Reduced encoding}
    \end{subfigure}
    \caption{Illustration of~\Cref{example:BVA-1}}
    \label{fig:example-biclique}
\end{figure}


More generally, let $\mathcal{C}$ and $\mathcal{D}$ be sets of clauses. Then,  using notation $\mathcal{C} \bowtie \mathcal{D}$ to denote the set of clauses of the form $(C_i \lor D_j)$ for $C_i \in \mathcal{C}$ and $D_j \in \mathcal{D}$, we define a \emph{BVA step} as follows:

\begin{definition}[BVA Step]
    Let $\varphi$ be a formula, and let $\mathcal{C}$ and $\mathcal{D}$ be sets of nonempty clauses such that $\mathcal{C} \bowtie \mathcal{D} \subseteq \varphi$, and $y \not\in \textsf{Var}(\varphi)$ be a fresh variable. Then, we say that the formula
    \[
        \varphi' := \left(\varphi \setminus (\mathcal{C} \bowtie \mathcal{D})\right) \cup
        \{(C_i \lor y) \mid C_i \in \mathcal{C}\} \cup \{(\overline{y} \lor D_j) \mid D_j \in \mathcal{D}\}
    \]
    can be derived from $\varphi$ in one \emph{BVA step}. We summarize this by writing $\varphi \bvastep \varphi'$.
\end{definition}

Note that if $\varphi \bvastep \varphi'$, then fully resolving on the introduced fresh variable $y$ on $\varphi'$ recovers $\varphi$; in other words, BVA is the opposite of variable elimination in the Davis--Putnam sense. BVA operates by the repeated application of BVA steps, so we define \emph{idealized BVA} to be the reflexive transitive closure of $\bvastep$, denoted $\bvachain$.

\begin{lemma}[\!\!{\cite[Thm. 1]{mantheyAutomatedReencodingBoolean2012}}]\label{lemma:bva-encoding}
    Let $f$ be a boolean function and $\varphi$ a CNF encoding of $f$ over the base variables $X$. If $\varphi \bvachain \varphi'$, then $\varphi'$ is a CNF encoding of $f$ over the base variables $X$.
\end{lemma}




Detecting patterns of the form $\mathcal{C} \bowtie \mathcal{D}$ for arbitrary sets $\mathcal{C}$ and $\mathcal{D}$ 
can be computationally expensive, and thus actual implementations of BVA restrict themselves to identifying subformulas $\mathcal{C} \bowtie \mathcal{D}$ in which $\mathcal{C}$ is a set of unit clauses~\cite{mantheyAutomatedReencodingBoolean2012, haberlandtEffectiveAuxiliaryVariables2023}. Note that each clause in $(C_i \lor D_j) \in \mathcal{C} \bowtie \mathcal{D}$ has size at most $|C_i| + |D_j|$, and thus in the case we focus on this paper, where $\varphi$ is a 2-CNF formula, we may assume that \emph{both} $\mathcal{C}$ and $\mathcal{D}$ are sets of unit clauses. Thus, for ease of notation we will write, e.g., $\mathcal{C} = \{x_1, x_2\}$ instead of $\mathcal{C} = \{(x_1), (x_2)\}$.

\section{BVA in Terms of Strict Polarized Rectifier Networks} \label{sec-bva-characterization}
Rectifier networks date back to the 1950s~\cite{lupanov}, and arise in several different contexts, including circuit complexity, automata theory, and graph theory~\cite{ivanBicliqueCoveringsRectifier2014,juknaComputationalComplexityGraphs2013, chistikov_et_al:LIPIcs.STACS.2017.23}. In general terms, they provide a way of representing the reachability relationship between vertices of a directed graph by introducing auxiliary vertices with the goal of reducing the total number of edges. This reduction in edge count is the graph-theoretic interpretation of the clause count reduction in~\Cref{fig:example-biclique}, and throughout this section, we will develop a precise correspondence between a graph-theoretic and clausal representation for 2-CNF encodings.

There are several ways of representing a 2-CNF formula as a graph, with the most common being the \emph{implication graph}, in which a clause $(\overline{p} \lor q)$ corresponds to directed edges $p \to q$ and $\overline{q} \to \overline{p}$. For our analysis of idealized BVA, having two vertices per variable and two edges per clause will be inconvenient, which motivates us to aim for a graph representation (a \emph{diagram}) in which every variable corresponds to a unique vertex, and every binary clause a unique edge. Binary clauses of mixed polarity, such as $(\overline{p} \lor q)$, are mapped to a directed edge $(p,q)$, while clauses $(\overline{p} \lor \overline{q})$ are mapped to an undirected edge $\{p,q\}$. But clauses of the form $(p \lor q)$ cannot be neatly represented in this framework. To that end, we define the following class of 2-CNF formulas:



\begin{definition}[Simple 2-CNF]
We say a 2-CNF formula $\varphi$ is \emph{simple} if:
\begin{enumerate}[label=(\alph*)]
    \item every strongly connected component of the implication graph of $\varphi$ contains exactly one literal, \label{four-cond}
    \item every clause has at least one negative literal (i.e., $\varphi$ is Horn), and \label{third-cond}
    \item for every pair of distinct clauses $C_1, C_2 \in \varphi$, $\textsf{Var}(C_1) \neq \textsf{Var}(C_2)$. \label{second-cond}
\end{enumerate}
\end{definition}

The goal of this section is to characterize how idealized BVA operates on simple 2-CNF formulas. There is no real loss of generality in restricting our attention to simple 2-CNF formulas, because any satisfiable 2-CNF formula can be converted to a simple one that is equivalent up to replacing some variables by their negations. This simplification is furthermore efficient, as formalized by the following proposition:
\begin{restatable}{proposition}{propsimpletogeneral} \label{prop:simple-to-general}
    If every simple 2-CNF formula on $n$ variables has a 2-CNF encoding with at most $f(n)$ clauses, computable in time $g(n)$, then every 2-CNF formula on $n$ variables has a 2-CNF encoding with at most $f(n) + O(n)$ clauses, also computable in time $O(g(n))$.
\end{restatable}
\begin{proof}[Proof sketch]
    Unsatisfiable formulas can be trivially reencoded into the empty clause, so suppose that $\varphi$ is a satisfiable 2-CNF formula we wish to reencode. Using equivalent literal substitution, we can reduce to a formula $\varphi'$ whose implication graph contains no strongly connected components of size greater than $1$; it suffices to reencode $\varphi'$, since the equivalences of literals can be appended afterward. Next, let $L$ be the set of literals $\ell$ such that $(\ell \lor x_j)$ and $(\ell \lor \overline{x_j})$ are both present in $\varphi'$ for some $x_j$, and note that $\varphi' \models \ell$ for $\ell \in L$, so $\varphi' \equiv \varphi' \land \bigwedge_{\ell \in L} \ell$. Then, applying unit propagation on $\varphi' \land \bigwedge_{\ell \in L} \ell$ yields a formula $\varphi''$ containing no pairs of clauses over the same two variables; it suffices to reencode $\varphi''$, since the unit clauses $\ell$ can be dealt with afterward. Finally, since $\varphi''$ is satisfiable and 2-CNF, it is Horn-renamable (see \Cref{lem:no-pos}), so we can reduce to a formula $\varphi'''$ obtained by replacing some variables by their negations throughout $\varphi''$, and an encoding for $\varphi'''$ can be transformed into an encoding for $\varphi''$ via the same replacement.
\end{proof}

Now, we introduce the graph-theoretic notions at the heart of our framework and explain their correspondence to the logical notions we study (see \Cref{tab:dictionary}). A virtue of working with simple 2-CNF formulas is that they have a convenient graph-theoretic representation using what we call \emph{diagrams} (a similar definition appears in \cite{allen2cnf}):
\begin{definition}[Diagram]
    A \emph{diagram} $G = (V,E)$ is a graph with both undirected and directed edges allowed (but no loops) and such that $|\{\{u, v\}, (u, v), (v, u)\} \cap E| \le 1$ for all $u,v \in V$. Given a simple 2-CNF formula $\varphi$, its \emph{associated diagram} $G_\varphi$ is defined by $V(G_\varphi) = \textsf{Var}(\varphi)$ and $E(G_\varphi) = \{\{x_i,x_j\} \mid (\overline{x_i} \lor \overline{x_j}) \in \varphi\} \cup \{(x_i,x_j) \mid (\overline{x_i} \lor x_j) \in \varphi\}$.
\end{definition}

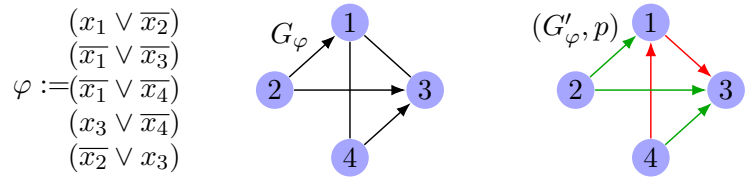
\begin{figure}[b]
    \centering
    \begin{tikzpicture}[ v/.style   ={circle, draw=none, fill=blue!35, minimum size=5mm, inner sep=0pt}]

        \node[] (phi) at (-1.1, -0.9) {$\varphi :=$};
        \def\vertDistance{0.45}
        \node[] (c1) at (0,0) {$(x_1 \lor \overline{x_2})$};
        \node[] (c2) at (0,-1*\vertDistance) {$(\overline{x_1} \lor \overline{x_3})$};
        \node[] (c3) at (0,-2*\vertDistance) {$(\overline{x_1} \lor \overline{x_4})$};
        \node[] (c4) at (0,-3*\vertDistance) {$(x_3 \lor \overline{x_4})$};
        \node[] (c5) at (0,-4*\vertDistance) {$(\overline{x_2} \lor x_3)$};

        \node[] (gphi) at (3.2-1, -0.2) {$G_\varphi$};
        \node[v] (v1) at (4-1, 0) {1};
        \node[v] (v2) at (3-1, -0.9) {2};
        \node[v] (v3) at (5-1, -0.9) {3};
        \node[v] (v4) at (4-1, -1.8) {4};

        \draw[line width=0.55pt] (v1) -- (v3);
        \draw[line width=0.55pt] (v1) -- (v4);

        \draw[line width=0.55pt, -Latex] (v2) -- (v1);
        \draw[line width=0.55pt, -Latex] (v4) -- (v3);
        \draw[line width=0.55pt, -Latex] (v2) -- (v3);

        \def\hdist{3.0}
         \node[] (gphip) at (3.2 + \hdist -0.2, -0.1) {$(G'_\varphi, p)$};
        \node[v] (v1p) at (4+ \hdist, 0) {1};
        \node[v] (v2p) at (3+ \hdist, -0.9) {2};
        \node[v] (v3p) at (5+ \hdist, -0.9) {3};
        \node[v] (v4p) at (4+ \hdist, -1.8) {4};

        \draw[line width=0.55pt, red, -Latex] (v1p) -- (v3p);
        \draw[line width=0.55pt, red, -Latex] (v4p) -- (v1p);

        \draw[line width=0.55pt, -Latex, \edgegreen] (v2p) -- (v1p);
        \draw[line width=0.55pt, -Latex, \edgegreen] (v4p) -- (v3p);
        \draw[line width=0.55pt, -Latex, \edgegreen] (v2p) -- (v3p);

    \end{tikzpicture}
    \caption{Illustration of a diagram and a polarized diagram for a simple 2-CNF. Edges with \emph{positive polarity} (i.e., $p(u, v) = +$) are colored green, whereas \emph{negative polarity} edges are colored red.}
    \label{fig:polarized-diagram}
\end{figure}

When we define polarized rectifier networks and their relation to diagrams, it will be convenient to work with \emph{polarized diagrams}, in which each undirected edge of a diagram is oriented arbitrarily (see \Cref{fig:polarized-diagram}):
\begin{definition}[Polarized Diagram]
    Given a diagram $G$ with undirected edges $E_1$ and directed edges $E_2$, a \emph{polarized diagram} of $G$ is a pair $(G',p)$ consisting of a directed graph $G'$ obtained from $G$ by orienting every edge in $E_1$ together with a function $p \colon E(G') \to \{-,+\}$ given by
    \(
        p(u,v) = \begin{cases}
            - \quad &\text{if} \ \{u,v\} \in E_1 \\
            + \quad &\text{if} \ (u,v) \in E_2.
        \end{cases}
    \)
\end{definition}
When BVA reencodes a formula containing a clause of the form $(\overline{x_i} \lor \overline{x_j})$, it can do so using a chain of implications either of the form $x_i \rightarrow y_1 \rightarrow \dots \rightarrow y_k \rightarrow \overline{x_j}$ or of the form $x_j \rightarrow y_1 \rightarrow \dots \rightarrow y_k \rightarrow \overline{x_i}$, where $y_1, \dots, y_k$ are auxiliary variables. In a polarized diagram, the orientation of the edge $\{x_i,x_j\}$ represents a choice of which of these alternatives idealized BVA may take. Concretely, if a clause $(\overline{x_i} \lor \overline{x_j})$ ever forms part of a BVA step $\mathcal{C} \bowtie \mathcal{D}$, we will orient it as $(x_i, x_j)$ if $\overline{x_i} \in \mathcal{C}$ and $\overline{x_j}  \in \mathcal{D}$, and as $(x_j, x_i)$ if $\overline{x_j} \in \mathcal{C}$ and $\overline{x_i}  \in \mathcal{D}$. If the clause is never part of a BVA step, then its edge orientation can be chosen arbitrarily.

Next, we extend polarized diagrams into \emph{polarized rectifier networks}, which are polarized diagrams with a distinguished set of \emph{auxiliary vertices}:


\begin{definition}[Polarized Rectifier Network]
    A \emph{polarized rectifier network} (PRN) $\mathcal{R}$ is a quadruple $(B, A, E, p)$, where $B$ and $A$ are disjoint sets of \emph{vertices}, $(B \sqcup A, E)$ is a directed graph without loops, and $p$ is a function $\{(u,v) \in E \mid v \in B\} \to \{-,+\}$. We say that $B$ and $A$ are the \emph{base} and \emph{auxiliary} vertices of $\mathcal{R}$ respectively. Given $u,v \in B$ (potentially $u = v$), a \emph{valid walk} from $u$ to $v$ in $\mathcal{R}$ is a sequence of vertices $\pi := (\pi_1, \ldots, \pi_k)$ for some $k \ge 2$ such that $(\pi_1,\pi_k) = (u,v)$, $(\pi_i, \pi_{i+1}) \in E$ for each $i \in [k-1]$, and $\pi_i \in A$ for each $i \in [2,k-1]$. We say that $\pi$ is \emph{positive} or \emph{negative} according to the value of $p(\pi_{k-1},\pi_k)$. When a valid walk $\pi$ is positive, we write $p(\pi) = +$, and $p(\pi) = -$ otherwise.
\end{definition}

Recall from \Cref{tab:dictionary} that PRNs correspond to simple 2-CNF formulas with auxiliary variables.\footnote{The correspondence here is not exact. Some simple 2-CNF formulas with auxiliary variables (namely, those containing clauses of the form $(\overline{y}_i \lor \overline{y_j})$, where $y_i$ and $y_j$ are auxiliary variables) do not correspond to polarized rectifier networks. However, our goal is to characterize encodings constructible by idealized BVA (\Cref{thm:characterization}), which never constructs clauses of that form.} If $\mathcal{R} = (B, A, E, p)$ is a PRN, then $B$ and $A$ correspond to base and auxiliary variables respectively. An edge $(u,v) \in E$ corresponds to an implication $u \rightarrow v$, unless $v$ is a base vertex and $p(u,v) = -$, in which case it corresponds to the implication $u \rightarrow \overline{v}$. A positive (respectively, negative) valid walk from $u$ to $v$ in $\mathcal{R}$ corresponds to a chain of implications $u \rightarrow y_2 \rightarrow \dots \rightarrow y_{k-1} \rightarrow v$ (respectively, $u \rightarrow y_2 \rightarrow \dots \rightarrow y_{k-1} \rightarrow \overline{v}$), where $y_2, \dots, y_{k-1}$ are auxiliary variables. Next, we define the translation from a PRN to a 2-CNF formula that makes this correspondence precise:
\begin{definition}[PRN to Encoding]
    Given a PRN $\mathcal{R} = (B, A, E, p)$, we define a 2-CNF formula $F_\mathcal{R}$ over the variables $B \sqcup A$ by mapping each edge $(u,v) \in E$ into a clause $C_{(u,v)}$ as follows: if $v \in B$ and $p(u,v) = -$, then $C_{(u,v)} := (\overline{u} \lor \overline{v})$; otherwise, $C_{(u,v)} := (\overline{u} \lor v)$. Finally, $F_\mathcal{R} := \bigwedge_{e \in E} C_e$.
\end{definition}

\begin{table}[t]
\centering
\caption{Translation between formulas/encodings and graph-theoretic notions}
\label{tab:dictionary}
\renewcommand{\arraystretch}{1.15}
\begin{tabularx}{\linewidth}{X|X}
\toprule
\textbf{Formulas and encodings} & \textbf{Graph-theoretic notions} \\ \midrule
Simple 2-CNF formula & Diagram \\ 
Simple 2-CNF formula with auxiliary variables & Polarized rectifier network \\ 
Implication chain $x_i \! \to \! y_2 \! \to \! \dots \! \to \! y_{k-1} \! \to \pm x_j$ & Valid walk $(x_1, y_1, \dots, y_{k-1}, x_j)$\\
A simple 2-CNF formula with auxiliary variables \emph{encodes} a simple 2-CNF formula & A polarized rectifier network \textit{realizes} a diagram \\
Encoding constructible by idealized BVA & Strict polarized rectifier network\\
\bottomrule
\end{tabularx}
\end{table}

Next, we define the relation between a PRN $\mathcal{R}$ and a diagram $G_\varphi$ that holds when $F_\mathcal{R}$ is an encoding of $\varphi$:
\begin{definition}[PRN Realizing a Diagram]
    Given a polarized diagram $(G,p)$, a PRN $\mathcal{R} = (V(G), A, E, p')$ \emph{realizes} $(G,p)$ if for every pair $u,v \in V(G)$ (potentially $u = v$), we have $(u,v) \in E(G)$ if and only if there is a valid walk $\pi$ from $u$ to $v$ in $\mathcal{R}$ such that $p'(\pi) = p(u, v)$. If $(G,p)$ is a polarized diagram of $G'$ and $\mathcal{R}$ realizes $(G,p)$, we also say that $\mathcal{R}$ \emph{realizes} $G'$.
\end{definition}
The previous definition is motivated by the fact that direct implications in a formula, which correspond to edges of a polarized diagram, must be captured by implication chains in an encoding, which correspond to valid walks in a PRN. Note that given a polarized diagram $(G,p)$, we can naturally define a PRN $\mathcal{R}_{(G,p)} = (V(G),\emptyset,E(G),p)$ that realizes $(G,p)$.

\begin{remark}\label{remark:varphi-eq-F}
    Let $\varphi$ be a simple 2-CNF formula, and let $(G, p)$ be a polarized diagram of $G_{\varphi}$. Then, $\varphi = F_{\mathcal{R}_{(G, p)}}$.
\end{remark}

It turns out that polarized rectifier networks represent a more general class of simple 2-CNF encodings than what can be achieved by idealized BVA. To that end, we define the following:
\begin{definition}[Strict Polarized Rectifier Network]
        Let $\mathcal{R} = (B,A,E,p)$ be a PRN. We say that $\mathcal{R}$ is a \emph{strict polarized rectifier network} (SPRN) if the following conditions are satisfied:
        \begin{enumerate}
            \item for every $\{u,v\} \in \binom{B}{2}$, $\mathcal{R}$ contains at most one valid walk from $u$ to $v$ or from $v$ to $u$;
            \item every vertex $u \in A$ belongs to some valid walk.
        \end{enumerate}
\end{definition}
The second condition of this definition rules out ``dangling'' auxiliary vertices, which serve no purpose and are never constructed by idealized BVA. The first condition is more interesting. It forbids us from realizing an implication between base vertices via multiple valid walks, something idealized BVA will never do. There is, however, no principled justification for this restriction; for some diagrams, the smallest SPRNs realizing them may be asymptotically larger than the smallest PRNs realizing them~\cite[Chapter~5]{linear-boolean}. This suggests a possible avenue for improving BVA by allowing it to construct multiple implication chains capturing an implication between base variables.

Our characterization theorem identifies which formulas can be obtained by idealized BVA from a simple 2-CNF formula. To state it, we define the following equivalence relation on formulas: Given formulas $\varphi$ and $\varphi'$ and a set of variables $W$, we write $\varphi \cong_W \varphi'$ if $\varphi'$ can be obtained from $\varphi$ by replacing some variables from $W$ by their negations wherever they appear. Naturally, replacing auxiliary variables by their negations is irrelevant for encodings: if $\varphi$ is an encoding of some boolean function $f$ with auxiliary variables $W$, and $\varphi \cong_W \varphi'$, then $\varphi'$ is also an encoding of $f$. We can now state our characterization theorem:

\begin{restatable}[BVA Characterization]{theorem}{bvacharacterization} \label{thm:characterization}
    Let $\varphi$ and $\varphi'$ be 2-CNF formulas, where $\varphi$ is simple. We have $\varphi \bvachain \varphi'$ if and only if there is an SPRN $\mathcal{R}$ realizing the diagram $G_{\varphi}$ with $F_\mathcal{R} \cong_W \varphi'$, where $W = \textsf{Var}(\varphi') \setminus \textsf{Var}(\varphi)$.
\end{restatable}


\begin{figure}[t]
\centering

\tikzset{
  v/.style   ={circle, draw=none, fill=blue!35, minimum size=5mm, inner sep=0pt, font=\small},
  hot/.style ={circle, draw=none, fill=orange!70, minimum size=5mm, inner sep=0pt, font=\small},
  >={Latex},
  e/.style   ={\edgegreen, line width=0.55pt},
  en/.style   ={black, line width=0.55pt},
  er/.style={red, line width=0.55}
}


\begin{subfigure}[t]{0.32\linewidth}
\centering
\begin{tikzpicture}[x=1.00cm,y=1.05cm]


\node[v]   (3) at (4,3.2) {1};
\node[v]   (4) at (1,3.2) {2};


\node[v]   (7) at (1,2) {4};
\node[v]   (6) at (4,2) {3};


\node[v]   (2) at (1,0) {6};
\node[v]   (1) at (4,0) {5};
\node[v]   (5) at (2,0) {7};

\node[v]   (0) at (2,-1) {8};

\draw[er,->] (3) -- (6);
\draw[e,->] (3) to[bend right=20] (1);
\draw[e,->] (3) -- (5);
\draw[e,->] (3) -- (0);
\draw[e,->] (3) -- (2);

\draw[er,->] (4) -- (6);
\draw[e,->] (4) -- (7);
\draw[e,->] (4) to[bend right=20] (2);
\draw[e,->] (4) -- (1);
\draw[e,->] (4) -- (5);
\draw[e,->] (4) -- (0);

\draw[e,->] (6) -- (7);
\draw[er,->] (6) -- (1);
\draw[e,->] (6) -- (5);

\draw[e,->] (7) -- (1);
\draw[e,->] (7) -- (2);
\draw[e,->] (7) -- (5);

\draw[e,->] (2) -- (5);
\draw[er,->] (2) -- (0);

\draw[e,->] (5) -- (0);

\end{tikzpicture}
\caption{$(G, p)$}
\label{fig:phi-polarized}
\end{subfigure}
\hfill
\begin{subfigure}[t]{0.33\linewidth}
\centering
\begin{tikzpicture}[x=1.00cm,y=1.05cm]

\node[v]   (3) at (4,3.2) {1};
\node[v]   (4) at (1,3.2) {2};

\node[hot] (8) at (3,2.7) {1};

\node[v]   (7) at (1,2) {4};
\node[v]   (6) at (4,2) {3};


\node[v]   (2) at (1,0) {6};
\node[v]   (1) at (4,0) {5};
\node[v]   (5) at (2,0) {7};

\node[v]   (0) at (2,-1) {8};

\draw[en,->] (3) -- (8);
\draw[en,->] (4) -- (8);
\draw[e,->] (4) -- (7);

\draw[er,->] (8) -- (6);
\draw[e,->] (8) -- (0);

\draw[e,->] (6) -- (7);
\draw[er,->] (6) -- (1);
\draw[e,->] (6) -- (5);

\draw[e,->] (8) -- (2);
\draw[e,->] (8) -- (1);
\draw[e,->] (8) -- (5);

\draw[e,->] (7) -- (2);
\draw[e,->] (7) -- (1);
\draw[e,->] (7) -- (5);


\draw[e,->] (2) -- (5);
\draw[er,->] (2) -- (0);

\draw[e,->] (5) -- (0);

\end{tikzpicture}
\caption{$\mathcal{R}_1$}
\label{fig:first-bva-step}
\end{subfigure}
\hfill
\begin{subfigure}[t]{0.33\linewidth}
\centering
\begin{tikzpicture}[x=1.00cm,y=1.05cm]

\node[v]   (3) at (4,3.2) {1};
\node[v]   (4) at (1,3.2) {2};

\node[hot] (8) at (3,2.7) {1};

\node[v]   (7) at (1,2) {4};
\node[v]   (6) at (4,2) {3};

\node[hot] (9) at (2,1) {2};

\node[v]   (2) at (1,0) {6};
\node[v]   (1) at (4,0) {5};
\node[v]   (5) at (2,0) {7};

\node[v]   (0) at (2,-1) {8};

\draw[en,->] (3) -- (8);
\draw[en,->] (4) -- (8);
\draw[e,->] (4) -- (7);

\draw[er,->] (8) -- (6);
\draw[en,->] (8) -- (9);
\draw[e,->] (8) -- (0);

\draw[e,->] (6) -- (7);
\draw[er,->] (6) -- (1);
\draw[e,->] (6) -- (5);

\draw[en,->] (7) -- (9);

\draw[e,->] (9) -- (2);
\draw[e,->] (9) -- (5);
\draw[e,->] (9) -- (1);

\draw[e,->] (2) -- (5);
\draw[er,->] (2) -- (0);

\draw[e,->] (5) -- (0);

\end{tikzpicture}
\caption{$\mathcal{R}_2$}
\label{fig:second-bva-step}
\end{subfigure}

\caption{Illustration of the BVA steps in \Cref{ex-bva}}
\end{figure}
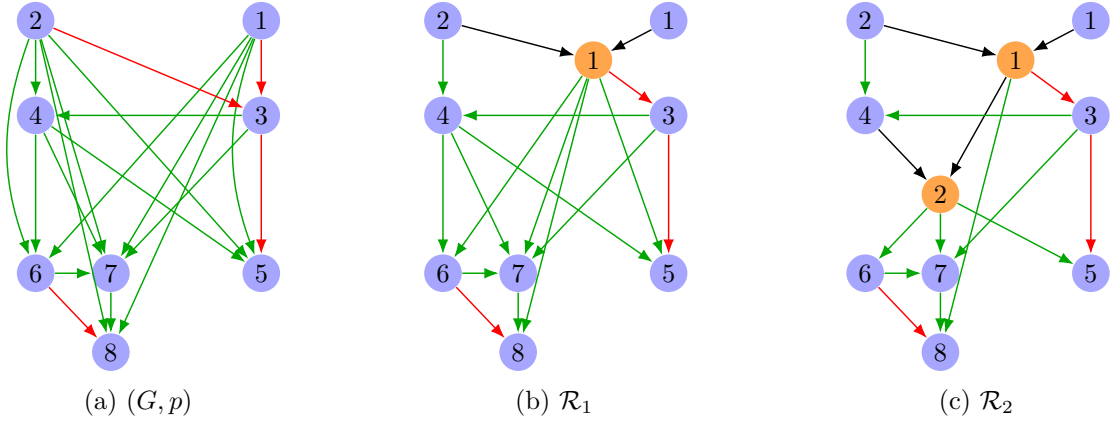

Since the proof of~\Cref{thm:characterization} (in \Cref{sec-sprn-appendix}) is quite technical, let us give some intuition for the result through the following example:
\begin{example} \label{ex-bva}
Consider the following simple 2-CNF formula:
\begin{align*}
   \varphi := \;& (\overline{x_1} \lor \overline{x_3})  \land (\overline{x_1} \lor x_5) \land (\overline{x_1} \lor x_6) \land (\overline{x_1} \lor x_7) \land (\overline{x_1} \lor x_8) 
    \land  (\overline{x_2} \lor \overline{x_3}) \land (\overline{x_2} \lor x_5) \\
    &\land  (\overline{x_2} \lor x_4) \land (\overline{x_2} \lor x_6) \land (\overline{x_2} \lor x_7) \land (\overline{x_2} \lor x_8) 
     \land (\overline{x_3} \lor \overline{x_5}) \land (\overline{x_3} \lor x_4) \land (\overline{x_3} \lor x_7) \\ 
     &\land (\overline{x_4} \lor x_5) \land (\overline{x_4} \lor x_6) \land (\overline{x_4} \lor x_7)
     \land (\overline{x_6} \lor x_7) \land (\overline{x_6} \lor \overline{x_8}) \land (\overline{x_7} \lor x_8).
\end{align*}

A polarized diagram for $\varphi$ is given in \Cref{fig:phi-polarized}. As our first BVA step, we take $\mathcal{C} = \{\overline{x_1}, \overline{x_2}\}$ and $\mathcal{D} = \{\overline{x_3}, x_5, x_6, x_7, x_8\}$, and replace $\mathcal{C} \bowtie \mathcal{D}$ by $\{(C_i \lor y_1) \mid C_i \in \mathcal{C}\} \cup \{(\overline{y_1} \lor D_j) \mid D_j \in \mathcal{D}\}$, where $y_1$ is a new auxiliary variable. After this BVA step, the formula corresponds to the rectifier network in \Cref{fig:first-bva-step}.
For the second BVA step, we take $\mathcal{C} = \{\overline{x_4}, \overline{y_1}\}$ and $\mathcal{D} = \{x_5, x_6, x_7\}$, and replace $\mathcal{C} \bowtie \mathcal{D}$ by $\{(C_i \lor y_2) \mid C_i \in \mathcal{C}\} \cup \{(\overline{y_2} \lor D_j) \mid D_j \in \mathcal{D}\}$, where $y_2$ is a new auxiliary variable. The resulting formula corresponds to the rectifier network in \Cref{fig:second-bva-step}.
\end{example}

\section{Encodings for General 2-CNF Formulas} \label{sec-2cnf-general}

As a first application of our characterization of idealized BVA (\Cref{thm:characterization}), we answer the following natural question: Given an arbitrary 2-CNF formula, how many clauses are in its smallest 2-CNF encoding constructible by idealized BVA? The answer, which depends on whether we use the simplification from \Cref{prop:simple-to-general}, is presented in \Cref{tab:bounds}.

\begin{table}[htbp]
\centering
\caption{Size of the smallest 2-CNF reencoding of a 2-CNF formula in the worst case}
\begin{tabular}{@{~~}l@{~~~~~~~~}c@{~~~~~~~~~~}c@{~~}}
\toprule
\textbf{} & \textbf{Lower bound} & \textbf{Upper bound} \\ \midrule
Idealized BVA & \makecell{$\displaystyle \left(1 - o(1)\right) \frac{n^2}{\lg n}$ \\[2ex] (\Cref{prop-bva-no-preprocessing-lb})} & \makecell{$\displaystyle \left(1 + o(1)\right) \frac{n^2}{\lg n}$ \\[2ex] (\Cref{thm:bva-no-preprocessing})} \\ \hline
Idealized BVA with simplification & \makecell{$\displaystyle \Big(\frac{\lg(3)}{4} - o(1)\Big) \frac{n^2}{\lg n}$ \\[2ex] (\Cref{prop-bva-with-preprocessing-lb})} & \makecell{$\displaystyle \Big(\frac{\lg(3)}{4} + o(1)\Big) \frac{n^2}{\lg n}$ \\[2ex] (\Cref{cor-bva-with-preprocessing})} \\ \hline
Arbitrary 2-CNF reencoding methods & \makecell{$\displaystyle \left(\frac{1}{4} - o(1)\right) \frac{n^2}{\lg n}$ \\[2ex] (\Cref{prop-monotone-lb})} & \makecell{$\displaystyle \Big(\frac{\lg(3)}{4} + o(1)\Big) \frac{n^2}{\lg n}$ \\[2ex] (\Cref{cor-bva-with-preprocessing})} \\ \bottomrule
\end{tabular}
\label{tab:bounds}
\end{table}

Notice that, although a 2-CNF formula can have $\Omega(n^2)$ clauses, it always has an encoding with $O(n^2 / \lg n)$ clauses, and such an encoding can be constructed by idealized BVA. Furthermore, there is no reencoding method that can improve on this by more than a factor of 4. On the other hand, the simplification from \Cref{prop:simple-to-general} yields an improved bound when compared to idealized BVA alone. In fact, \Cref{cor-bva-with-preprocessing} proves that there is an implementation of idealized BVA with simplification that achieves the stated bound whose runtime is $O(n^2)$; note that the time complexity of \textsf{factor}, the BVA implementation used in \textsf{CaDiCaL} and \textsf{Kissat}, is at least $\Omega(n^3)$.

The upper bounds in \Cref{tab:bounds} make essential use of our characterization of idealized BVA in terms of SPRNs. Specifically, \Cref{thm:bva-no-preprocessing} and~\Cref{cor-bva-with-preprocessing} employ a technique used by Nechiporuk~\cite{nechiporuk} to construct small rectifier networks for bipartite graphs; we generalize his result to our setting of polarized rectifier networks, following the line of work of Krapivin et al.~\cite{stoc}, which used similar techniques to extend results about \emph{biclique partitions} (i.e., strict rectifier networks without edges between auxiliary vertices) from bipartite graphs to arbitrary graphs.

\subsection{Upper Bound for BVA With Simplification}
To prove \Cref{cor-bva-with-preprocessing}, we start by proving a similar result for simple 2-CNF formulas:
\begin{restatable}[Formal version of \Cref{thm-2cnf}]{theorem}{thmnechiporukdiagram} \label{thm:nechiporuk-diagram}
    For every simple 2-CNF formula $\varphi$ on $n$ variables, there is a 2-CNF formula $\varphi'$ such that $\varphi \bvachain \varphi'$ and $|\varphi'| \le \left(\frac{\lg(3)}{4} + o(1)\right) \frac{n^2}{\lg n}$.
\end{restatable}
\begin{proof}[Proof sketch]
    Let $G_\varphi$ be the associated diagram of $\varphi$. By \Cref{thm:characterization}, it suffices to build an SPRN $\mathcal{R}$ realizing $G_\varphi$ with at most $\left(\frac{\lg(3)}{4} + o(1)\right) \frac{n^2}{\lg n}$ edges. Since $\varphi$ is simple, the directed part of $G_\varphi$ is acyclic, so we can topologically order the vertices $v_1, \dots, v_n$ so that every directed edge $(v_i,v_j)$ satisfies $i < j$. We create a polarized diagram $(G,p)$ of $G_\varphi$ by orienting each $\{v_i,v_j\} \in E(G_\varphi)$ to be $(v_i,v_j)$, where $i < j$. Our SPRN $\mathcal{R} = (V(G), A, E, p')$ will be a realization of $(G,p)$ and thus of $G_\varphi$.

    Let $\log_3 x = \frac{\lg x}{\lg 3}$. Let $q = \lfloor n / \log_3^2 n\rfloor$, and partition $V(G)$ into blocks $B_1, \ldots, B_{\lceil n/q\rceil}$ of size at most $q$; specifically, let $B_i = \{v_{(i-1)\cdot q + 1}, \dots, v_{i\cdot q}\}$ or, if $i \cdot q > n$, a truncation thereof. Within each block $B_i$, for each pair of vertices $\{u, v\} \in \binom{B_i}{2}$ create in $\mathcal{R}$ an auxiliary vertex  $z_{\{u, v\}}$, and create edges $(u,z_{\{u, v\}})$ and $(v,z_{\{u, v\}})$. Now, let $r = \lfloor \log_3 n  - 3\log_3 \log_3 n\rfloor$, and partition as well $V(G)$ into parts  $P_1, \ldots, P_{\lceil n/r\rceil}$ of size at most $r$; specifically, let $P_i = \{v_{(i-1)\cdot r + 1}, \dots, v_{i\cdot r}\}$ or, if $i \cdot r > n$, a truncation thereof. For each part $P_i$ and edge $(u, v) \in E(G) \cap (P_i \times P_i)$, create in $\mathcal{R}$ a directed edge $(u,v)$ with $p'(u,v) = p(u,v)$.

    Now, for each part $P_i$ and each function $S \colon P_i \to \{-,0,+\}$, create in $\mathcal{R}$ an auxiliary vertex $y_S$, and create a directed edge $(y_S,v)$ with $p'(y_S,v) = -$ for every $v \in P_i$ satisfying $S(v) = -$, and create a directed edge $(y_S,v)$ with $p'(y_S,v) = +$ for every $v \in P_i$ satisfying $S(v) = +$. Given a vertex $v \in V(G)$, let $N^\pm(v) = \{w \in V(G) \mid (v,w) \in E(G) \ \text{and} \ p(v,w) = \pm\}$. Then, for each $S \colon P_i \to \{-,0,+\}$, let 
    \[A(S) := \{ v \in P_j \mid j < i, \ N^-(v) \cap P_i = S^{-1}(-), \ \text{and} \ N^+(v) \cap P_i = S^{-1}(+)\}.
    \]
    Finally, for each $\ell \in [ \lceil n/q \rceil]$, let $A_\ell(S) := A(S) \cap B_\ell$, and let $a_1, \ldots, a_k$ be the elements of $A_\ell(S)$ in an arbitrary order. Then, create in $\mathcal{R}$ directed edges $(z_{\{a_{2j-1}, a_{2j}\}}, y_S)$ for $j \in [\lfloor k/2\rfloor]$, and if $k$ is odd, create the directed edge $(a_k,y_S)$. If $\mathcal{R}$ contains any auxiliary vertices not belonging to a valid walk, then delete them and their incident edges.

    This completes the construction of $\mathcal{R}$. The full proof in the appendix proves that $\mathcal{R}$ is a strict rectifier network for $(G,p)$ and has at most $\left(\frac{\lg(3)}{4} + o(1)\right) \frac{n^2}{\lg n}$ edges.
\end{proof}
\noindent
Together with the simplification from \Cref{prop:simple-to-general}, we obtain the following:
\begin{corollary} \label{cor-bva-with-preprocessing}
    Every 2-CNF formula $\varphi$ on $n$ variables has a 2-CNF encoding $\varphi'$ such that $|\varphi'| \le (\frac{\lg(3)}{4} + o(1)) \frac{n^2}{\lg n}$, and moreover, such an encoding can be computed in time $O(n^2)$.
\end{corollary}
\begin{proof}
    It suffices to combine~\Cref{prop:simple-to-general} with~\Cref{thm:nechiporuk-diagram}, and observe from the previous proof that the construction can be carried out in time $O(n^2)$.\footnote{A similar construction, with pseudocode, appears in~\cite[Algorithm 1]{stoc}.}
\end{proof}

\subsection{Lower Bound for BVA With Simplification}

Using an information-theoretic argument, we now show that the bound from \Cref{thm:nechiporuk-diagram} is sharp. We start by bounding the number of 2-CNF formulas with at most $m$ clauses and the number of simple 2-CNF formulas on $n$ variables.

\begin{restatable}{lemma}{lemcountformulas} \label{lem-count-formulas}
    The number of 2-CNF formulas with at most $m$ clauses is at most $2^{(1+o(1)) m \cdot \lg m}$.
\end{restatable}

\begin{restatable}{lemma}{lemcountsimple} \label{lem-count-simple}
    The number of simple 2-CNF formulas on $n$ variables is at least $3^{(1/2-o(1)) n^2}$.
\end{restatable}

\noindent
Now, we can combine these two lemmas to prove the lower bound:

\begin{proposition} \label{prop-bva-with-preprocessing-lb}
    There is a simple 2-CNF formula $\varphi$ on $n$ variables such that, for every $\varphi'$ with $\varphi \bvachain \varphi'$, we have $|\varphi'| \ge \left(\frac{\lg(3)}{4} - o(1)\right) \frac{n^2}{\lg n}$.
\end{proposition}
\begin{proof}
    Let $g(m)$ be the number of 2-CNF formulas with at most $m$ clauses, and let $h(n)$ be the number of simple 2-CNF formulas on $n$ variables. By \Cref{lem-count-formulas,lem-count-simple}, $\lg g(m) \le (1+o(1)) m \cdot \lg m$ and $\lg h(n) \ge (\lg(3)/2 - o(1)) n^2$.

    If $\varphi \bvachain \varphi'$, then performing variable elimination (in the Davis--Putnam sense) on the auxiliary variables in $\varphi'$ yields $\varphi$. In particular, $\varphi'$ has enough information to uniquely reconstruct $\varphi$. Thus, if for every simple 2-CNF formula $\varphi$ on $n$ variables, there is a 2-CNF formula $\varphi'$ with at most $m$ clauses such that $\varphi \bvachain \varphi'$, then $g(m) \ge h(n)$, and consequently $\lg g(m) \geq \lg h(n)$.
    If $m \le (r+o(1))n^2/\lg n$ for some constant $r$, then $\lg m \le (2+o(1))\lg n$, and thus
    \[
        (\lg(3)/2 - o(1)) n^2 \le \lg h(n) \le \lg g(m) \le (1+o(1)) m \cdot \lg m \le (2r+o(1)) n^2,
    \]
    from where $r \ge \lg(3)/4$; that is, $m \ge (\frac{\lg(3)}{4} - o(1)) n^2/\lg n$.
\end{proof}

\subsection{Monotone Formulas}

The proof of \Cref{prop-bva-with-preprocessing-lb} exploits the fact that idealized BVA reencodes distinct simple 2-CNF formulas into distinct encodings. To get a lower bound for arbitrary reencoding methods, which may not have this property, we focus on monotone 2-CNF formulas:
\begin{restatable}{proposition}{lowerbound} \label{prop-monotone-lb}
    There is a monotone 2-CNF formula with $n$ variables whose smallest 2-CNF encoding has at least $(\frac{1}{4}-o(1)) \frac{n^2}{\lg n}$ clauses.
\end{restatable}
\noindent
In fact, idealized BVA can match this lower bound:
\begin{restatable}[Formal version of \Cref{thm-2cnf-monotone}]{theorem}{idealizedbvamonotone} \label{thm:idealized-bva-monotone}
    For every monotone 2-CNF formula $\varphi$ on $n$ variables, there is a 2-CNF formula $\varphi'$ such that $\varphi \bvachain \varphi'$ and $|\varphi'| \le \left(\frac{1}{4} + o(1)\right) \frac{n^2}{\lg n}$.
\end{restatable}
\noindent
This strengthens a result of Subercaseaux~\cite[Theorem~10]{subercaseaux2025asymptoticallysmallerencodingsgraph}, who proved that every monotone 2-CNF formula has a 2-CNF encoding with $O(n^2 / \lg n)$ clauses.

Allen~\cite{allen2cnf} proved the remarkable result that almost all 2-CNF functions are monotone after possibly negating some input variables (i.e., \emph{unate}). Together with \Cref{thm:idealized-bva-monotone}, this implies that almost every 2-CNF function on $n$ variables has a 2-CNF encoding with at most $(\frac{1}{4} + o(1)) n^2/\lg n$ clauses, matching the lower bound from \Cref{prop-monotone-lb}. On this basis, we conjecture that this bound holds for \emph{all} 2-CNF functions, which would imply that there is a reencoding method improving idealized BVA with simplification:
\begin{conjecture}
    Every 2-CNF formula $\varphi$ on $n$ variables has a 2-CNF encoding $\varphi'$ such that $|\varphi'| \le (\frac{1}{4} + o(1)) \frac{n^2}{\lg n}$.
\end{conjecture}

\section{Encodings for \textsf{AtMostOne}} \label{sec-amo}
The previous section showed that idealized BVA performs well from the perspective of worst-case analysis. But many 2-CNF formulas that arise in practice are highly structured and can therefore be encoded with far fewer than $\Theta(n^2 / \lg n)$ clauses. Our framework also allows us to analyze the operation of idealized BVA on specific 2-CNF formulas. We demonstrate this by studying encodings of the most structured 2-CNF formula of all, one which also frequently arises in practice, namely the at-most-one constraint:
\[
    \textsf{AtMostOne}(x_1,\dots,x_n) \coloneqq \bigwedge_{1 \le i < j \le n} (\overline{x_i} \lor \overline{x_j}).
\]
Actual implementations of BVA reencode this formula into one with $3n-6$ clauses~\cite{mantheyAutomatedReencodingBoolean2012}; while this has been empirically observed before, we formalize and prove this fact in \Cref{sec-amo-in-practice}. By \Cref{thm:characterization}, a BVA-constructible encoding for this formula corresponds to an SPRN realizing $K_n$. \Cref{fig:optimal-amo} depicts the corresponding SPRN for the encoding with $3n-6$ edges.

\begin{figure}[ht]
    \centering
    \begin{tikzpicture}[scale=1.2]

\tikzset{
  basenode/.style   ={circle, draw=none, fill=blue!35, minimum size=5mm, inner sep=0pt, font=\small},
  auxnode/.style    ={circle, draw=none, fill=orange!70, minimum size=5mm, inner sep=0pt, font=\small},
  >={Latex},
  e/.style          ={blue, line width=0.55pt},
  en/.style         ={black, line width=0.55pt},
  er/.style         ={red, line width=0.55pt}
}


\node[basenode] (x1)  at (0, -2) {$1$};

\node[basenode] (x2)  at (1, 0) {$2$};
\node[basenode] (x3)  at (1, -1.25) {$3$};
\node[auxnode]  (y1)  at (2, -2) {$1$};

\node[basenode] (x4)  at (3, 0) {$4$};
\node[basenode] (x5)  at (3, -1.25) {$5$};
\node[auxnode]  (y2)  at (4, -2) {$2$};

\node[basenode] (x6)  at (5, 0) {$6$};
\node[basenode] (x7)  at (5, -1.25) {$7$};
\node[auxnode]  (y3)  at (6, -2) {$3$};

\node (dots1) at (7.25, 0) {$\cdots$};
\node (dots2) at (7.25, -1.25) {$\cdots$};
\node (dots3) at (7.25, -2) {$\cdots$};

\node[auxnode]  (yk)  at (8.5, -2) {\scalebox{0.65}{$\left\lfloor\!\tfrac{n-3}{2}\!\right\rfloor$}};
\node[basenode] (xn2) at (9.5, 0) {\scalebox{0.65}{$n\!-\!2$}};
\node[basenode] (xn1) at (9.5, -1.25) {\scalebox{0.65}{${n\!-\!1}$}};
\node[basenode] (xn)  at (10.5, -2) {${n}$};


\draw[er, ->] (x1) -- (x2);
\draw[er, ->] (x1) -- (x3);
\draw[en, ->] (x1) -- (y1);
\draw[er, ->] (x2) -- (x3);
\draw[en, ->] (x2) -- (y1);
\draw[en, ->] (x3) -- (y1);

\draw[er, ->] (y1) -- (x4);
\draw[er, ->] (y1) -- (x5);
\draw[en, ->] (y1) -- (y2);
\draw[er, ->] (x4) -- (x5);
\draw[en, ->] (x4) -- (y2);
\draw[en, ->] (x5) -- (y2);

\draw[er, ->] (y2) -- (x6);
\draw[er, ->] (y2) -- (x7);
\draw[en, ->] (y2) -- (y3);
\draw[er, ->] (x6) -- (x7);
\draw[en, ->] (x6) -- (y3);
\draw[en, ->] (x7) -- (y3);

\draw[er, ->] (yk) -- (xn2);
\draw[er, ->] (yk) -- (xn1);
\draw[er, ->] (yk) -- (xn);
\draw[er, ->] (xn2) -- (xn1);
\draw[er, ->] (xn2) -- (xn);
\draw[er, ->] (xn1) -- (xn);

\draw[er, dashed, ->] (y3) -- (6.8, -0.3);

\draw[en, dashed, ->] (y3) -- (6.9, -2);

\draw[er, dashed, ->] (y3) -- (6.7, -1.45);
\draw[en, dashed, ->] (7.7, -0.3) -- (yk);
\draw[en, dashed, ->] (7.8, -1.45) -- (yk);

\draw[en, dashed, ->] (7.6,-2.0) -- (yk);

\end{tikzpicture}
    \caption{SPRN realizing $K_n$ with $3n-6$ edges}
    \label{fig:optimal-amo}
\end{figure}
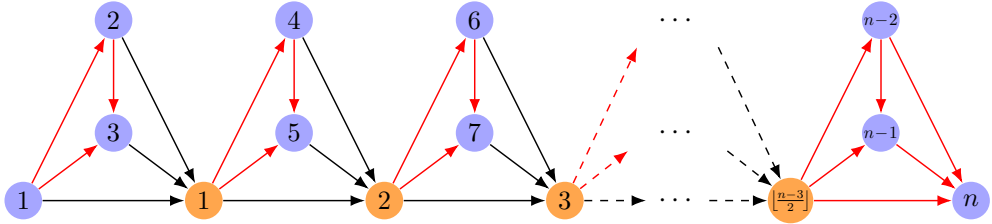

Using \Cref{thm:characterization}, we prove that $3n-6$ clauses is the smallest possible encoding for this constraint achievable by idealized BVA:
\begin{theorem}[Formal version of \Cref{thm:amo}] \label{thm:amo-formal}
    For every $\varphi$ with $\textsf{AtMostOne}(x_1,\dots,x_n) \bvachain \varphi$, we have $|\varphi| \ge 3n-6$.
\end{theorem}
\noindent
Note that there are 2-CNF encodings for $\textsf{AtMostOne}(x_1,\dots,x_n)$ using $2n + o(n)$ clauses~\cite{chen2010new}, and it was previously unknown whether idealized BVA could construct these. Thus, for a reencoding method to encode $\textsf{AtMostOne}(x_1,\dots,x_n)$ using fewer than $3n-6$ clauses, it must be able to create encodings that cannot be expressed as SPRNs.

We start with some lemmas that will help us establish properties of a minimal counterexample to \Cref{thm:amo-formal}.

\begin{restatable}{lemma}{lemcontract} \label{lem-contract}
    Let $\mathcal{R}$ be an SPRN with $m$ edges realizing a diagram $G$, and let $y$ be an auxiliary vertex. If the in-degree or out-degree of $y$ is 1, then there is an SPRN realizing $G$ with at most $m-1$ edges.
\end{restatable}
\begin{proof}[Proof sketch]
    If $(y',y)$ is the unique incoming edge to $y$, or $(y,y')$ is the unique outgoing edge from $y$, then we can contract this edge to get an SPRN with one fewer edge.
\end{proof}

\begin{lemma} \label{lem-deg-3}
    Let $\mathcal{R}$ be an SPRN realizing $K_n$ with $m$ edges for some $n \ge 4$, and let $x$ be a base vertex. If the total degree of $x$ is at least 3, then there is an SPRN realizing $K_{n-1}$ with at most $m-3$ edges.
\end{lemma}
\begin{proof}
    Let $\mathcal{R}'$ be obtained from $\mathcal{R}$ by deleting $x$ and its incident edges; also if $\mathcal{R}'$ contains any auxiliary vertices not belonging to a valid walk, then delete them and their incident edges. Then, $\mathcal{R}'$ is an SPRN realizing $K_{n-1}$ and has at most $m-3$ edges.
\end{proof}

\begin{restatable}{lemma}{lemdegone} \label{lem-deg-1}
    Let $\mathcal{R}$ be an SPRN realizing $K_n$ with $m$ edges for some $n \ge 4$, and let $x$ be a base vertex. If the total degree of $x$ is exactly 1, then there is an SPRN realizing $K_n$ with at most $m-1$ edges.
\end{restatable}
\begin{proof}[Proof sketch]
    If the total degree of $x$ is exactly 1, then its unique neighbor must be an auxiliary vertex $y$. Without loss of generality, $(x,y)$ is the edge in $\mathcal{R}$ connecting these two vertices. Then, $y$ has in-degree exactly 1, so the conclusion follows by \Cref{lem-contract}.
\end{proof}

\begin{lemma} \label{lem-deg-2}
    Let $\mathcal{R}$ be an SPRN realizing $K_n$ with $m$ edges for some $n \ge 4$, and let $x_1$ and $x_2$ be base vertices. If the total degrees of $x_1$ and $x_2$ in $\mathcal{R}$ are both 2 and there is an edge between $x_1$ and $x_2$, then there is an SPRN realizing $K_{n-1}$ with at most $m-3$ edges.
\end{lemma}
\begin{proof}
    Let $\mathcal{R}'$ be obtained from $\mathcal{R}$ by deleting $x_1$ and its two incident edges; also if $\mathcal{R}'$ contains any auxiliary vertices not belonging to a valid walk, then delete them and their incident edges. Then $\mathcal{R}'$ is an SPRN realizing $K_{n-1}$ and has at most $m-2$ edges. In $\mathcal{R}'$, the total degree of $x_2$ is exactly 1, so there is an SPRN realizing $K_{n-1}$ with at most $m-3$ edges by \Cref{lem-deg-1}.
\end{proof}

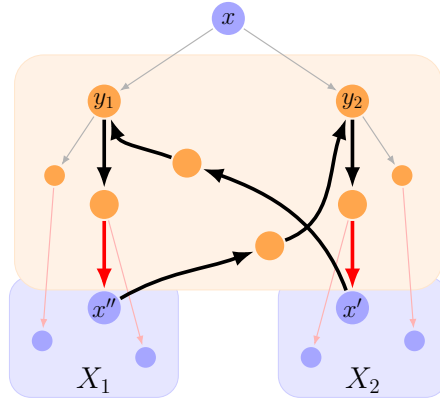
\begin{figure}
    \centering
    \scalebox{0.55}{
    \begin{tikzpicture}[
    scale=1.0,
    xnode/.style={circle,  fill=blue!35, thick, minimum size=0.5cm, inner sep=0pt},
    ynode/.style={circle,  fill=orange!70, thick, minimum size=0.7cm, inner sep=0pt},
    norm_edge/.style={-Latex, draw=gray!60, thick, shorten >=1pt, shorten <=1pt},
    hi_edge/.style={-Latex, draw=black, line width=2.5pt, shorten >=1pt, shorten <=1pt},
    key_node/.style={minimum size=0.8cm}
]


    \node[xnode, key_node] (x) at (0, 7) {\LARGE $x$};

    \node[ynode, key_node] (y1) at (-3, 5) {\LARGE  $y_1$};
    \node[ynode, key_node] (y2) at (3, 5) {\LARGE $y_2$};

    \node[xnode, key_node] (xpp) at (-3, 0) {\LARGE $x''$};
    \node[xnode] (x1a) at (-4.5, -0.8) {};
    \node[xnode] (x1b) at (-2, -1.2) {};

    \node[xnode, key_node] (xp) at (3, 0) {\LARGE $x'$};
    \node[xnode] (x2a) at (4.5, -0.8) {};
    \node[xnode] (x2b) at (2, -1.2) {};

    \node[ynode] (y_down1) at (-3, 2.5) {};
    \node[ynode] (y_down2) at (3, 2.5) {};
    \node[ynode] (y_cross1) at (-1, 3.5) {};
    \node[ynode] (y_cross2) at (1, 1.5) {};

    \node[ynode, minimum size=0.5cm] (y_fill1) at (-4.2, 3.2) {};
    \node[ynode, minimum size=0.5cm] (y_fill2) at (4.2, 3.2) {};

    \begin{pgfonlayer}{background}
        \node[fill=blue!10, draw=blue!20, rounded corners=15pt, fit=(xpp)(x1a)(x1b), inner sep=15pt, yshift=-5pt] { };

        \node[]  at (-3.25, -1.75) {\huge $X_1$};

        \node[fill=blue!10, draw=blue!20, rounded corners=15pt, fit=(xp)(x2a)(x2b), inner sep=15pt, yshift=-5pt] {};

         \node[]  at (3.25, -1.75) {\huge $X_2$};
        
        \node[fill=orange!10, draw=orange!20, rounded corners=20pt, fit=(y1)(y2)(y_cross2)(y_fill1)(y_fill2), inner sep=20pt] {};
    \end{pgfonlayer}


    \draw[norm_edge] (x) -- (y1);
    \draw[norm_edge] (x) -- (y2);

    \draw[norm_edge] (y1) -- (y_fill1);
    \draw[norm_edge,red!30] (y_fill1) -- (x1a);
    \draw[norm_edge,red!30] (y_down1) -- (x1b);
    \draw[norm_edge] (y2) -- (y_fill2);
    \draw[norm_edge, red!30] (y_fill2) -- (x2a);
    \draw[norm_edge, red!30] (y_down2) -- (x2b);


    \draw[hi_edge] (xp) to[out=110, in=-20] (y_cross1);
    \draw[hi_edge] (y_cross1) to[out=160, in=-70] (y1);

    \draw[hi_edge] (y1) -- (y_down1);
    \draw[hi_edge, red] (y_down1) -- (xpp);

    \draw[hi_edge] (xpp) to[out=30, in=200] (y_cross2);
    \draw[hi_edge] (y_cross2) to[out=20, in=250] (y2);

    \draw[hi_edge] (y2) -- (y_down2);
    \draw[hi_edge, red] (y_down2) -- (xp);

\end{tikzpicture}
    }
    \caption{Illustration for the proof of~\Cref{thm:amo-formal}}
    \label{fig:amo-lb}
\end{figure}

We are now ready to prove \Cref{thm:amo-formal}. Given a minimal counterexample (an SPRN realizing $K_n$ with fewer than $3n-6$ edges), we derive a contradiction by showing that the SPRN is not strict. Specifically, we identify two base vertices $x'$ and $x''$ such that we have a valid walk from $x'$ to $x''$ and a valid walk from $x''$ to $x'$. See \Cref{fig:amo-lb} for an illustration of the proof.

\begin{proof}[Proof of \Cref{thm:amo-formal}]
    Let $\mathcal{R}$ be an SPRN realizing $K_n$ with the minimum number of edges. By \Cref{thm:characterization}, it suffices to show that $\mathcal{R}$ has at least $3n-6$ edges. We may assume that $n \ge 3$. The proof is by induction on $n$. The base case, $n=3$, is easy to check. If $n \ge 4$, and there is some base vertex whose total degree is at least 3, then we are done by \Cref{lem-deg-3}. So assume that every base vertex has total degree at most 2. By \Cref{lem-deg-1}, every base vertex has total degree exactly 2. By \Cref{lem-deg-2}, we may assume that there are no edges between two base vertices.

    From these assumptions, we will arrive at a contradiction. Fix a base vertex $x$, and let $y_1$ and $y_2$ be its two adjacent auxiliary vertices. Suppose that the edges incident to $x$ are $(x,y_1)$ and $(x,y_2)$; if either edge goes in the other direction, the rest of the argument proceeds with only minor modifications. For $i \in \{1,2\}$, let $X_i$ be the elements of $V(K_n) \setminus \{x\}$ reachable by a valid walk starting with $(x,y_i)$. Then, since $\mathcal{R}$ is strict, $X_1 \cup X_2$ is a partition of $V(K_n) \setminus \{x\}$. By \Cref{lem-contract}, there is a base variable $x' \in V(K_n) \setminus \{x\}$ with a valid walk to $y_1$ and a base variable $x'' \in V(K_n) \setminus \{x\}$ with a valid walk to $y_2$. We must have $x' \in X_2$, or else there would be a valid walk $(x',\dots,y_1,\dots,x')$, which would contradict our assumption that $\mathcal{R}$ is an SPRN realizing $K_n$. Similarly, $x'' \in X_1$. But then, we have valid walks $(x',\dots,y_1,\dots,x'')$ and $(x'',\dots,y_2,\dots,x')$, which contradicts our assumption that $\mathcal{R}$ is strict.
\end{proof}






\section{Discussion and Experimental Results}\label{sec:discussion}

We have developed a theoretical framework to better understand Bounded Variable Addition (BVA) as a preprocessing technique, especially when accompanied by simplification techniques such as equivalent literal substitution and failed literal elimination.

We have focused on characterizing BVA's behavior and reencoding potential on the 2-CNF fragment of formulas. This not only simplifies the theoretical analysis, but is also justified by practical considerations. Namely, structured formulas are heavily biased towards narrow clauses: analyzing the 2025 SAT Competition, over 60\% of clauses have width 2 (avg. width is about 2.67, and fewer than 2\% have width greater than $4$). When \textsf{factor}, the state-of-the-art implementation of BVA, is run on these formulas, over 70\% of the clauses saved by the reencoding correspond to binary clauses, and fewer than $4\%$ correspond to width greater than $3$. Thus, width $3$ is arguably the only remaining case with practical applicability.

Generalizing our framework from 2-CNF to $k$-CNF will require reasoning about hypergraphs instead of graphs. Some initial work in this direction was recently done by Krapivin et al.~\cite{stoc}, who studied $k$-partite decompositions of $k$-uniform graphs, which can be seen as depth-2 rectifier networks for hypergraphs. In particular, their work implies that every monotone $k$-CNF formula on $n$ variables can be encoded with at most $(\frac{1}{k!} + o_k(1))n^{k}/k!$ clauses. We also remark that Allen's result that almost all 2-CNF functions are unate has recently been generalized to $k$-CNF functions~\cite{unate}.

Within the 2-CNF fragment, we have both established results for very structured formulas, such as encodings of the at-most-one constraint, as well as for general unstructured (e.g., random) formulas. For unstructured formulas, we have established sharp bounds on the reencoding potential of idealized BVA. While our main focus has been theoretical, we leverage an efficient algorithm of Krapivin et al.~\cite{stoc} for computing biclique partitions to implement a tool, \textsf{BiVA} (Biclique Variable Addition), which reencodes monotone 2-CNF fragments of formulas. Since \textsf{BiVA} is optimized for worst-case formulas, we do not aim to compete with other BVA implementations on structured formulas. Thus, we benchmark on formulas encoding the independent set problem on $G(n, \tfrac{1}{2})$-random graphs.

We benchmark against the \textsf{BVA} implementation from~\cite{mantheyAutomatedReencodingBoolean2012} (available at \url{https://fmv.jku.at/bva/}),
\textsf{factor} from \textsf{Kissat} (v4.0.4), and interestingly, we consider combinations \textsf{BiVA}+\textsf{BVA} and \textsf{BiVA}+\textsf{factor}, meaning that the reencoded output of \textsf{BiVA} is fed for a second round of reencoding (in contrast, if \textsf{BVA} or \textsf{factor} are applied to completion, then running \textsf{BiVA} on top has no effect). We do not include \textsf{SBVA} in our experiments since its exploitation of structure was not beneficial on random formulas, and the reencoding time was significantly longer than for the other methods.

First, as shown in~\Cref{fig:clause-compare}, the reduction in clauses is comparable between all methods, except for a 5--15\% margin by which \textsf{BiVA} is worse; however, this is almost entirely compensated by the combination with $\textsf{BVA}$ or $\textsf{factor}$. In fact, as shown in~\Cref{fig:runtime-both}, these combinations run much more efficiently than $\textsf{BVA}/\textsf{factor}$ in isolation. This can be explained by the fact that their runtimes depend on the amount of compression they achieve, and thus they run much faster on formulas that are already partially compressed. 
Finally, as shown in~\Cref{fig:aux-compare}, \textsf{BiVA} creates significantly fewer auxiliary variables, and unfortunately, its combination with $\textsf{BVA}/\textsf{factor}$ generates the most in aggregate.
More experimental details are provided in~\Cref{sec:experimemtal-details}.


\subparagraph{Future directions.} Our characterization of idealized BVA reveals not only its expressive power but also specific limitations that could be overcome by future reencoding methods. For example, one of BVA's limitations arises from its \emph{strictness}, which algorithmically corresponds to the fact that upon reencoding a biclique, BVA removes the original edges immediately, preventing any future reusage of them. This is necessary (but not sufficient) for automatically deriving the product encoding for \textsf{AtMostOne}. Given our promising preliminary experimental results, another direction is to develop a \textsf{BiVA}-inspired implementation that can exploit the structure present in real-world formulas. One step in this direction was accomplished by Krapivin et al.~\cite{stoc}, who showed how to efficiently construct biclique partitions 
that exploit the edge density of the input graph.



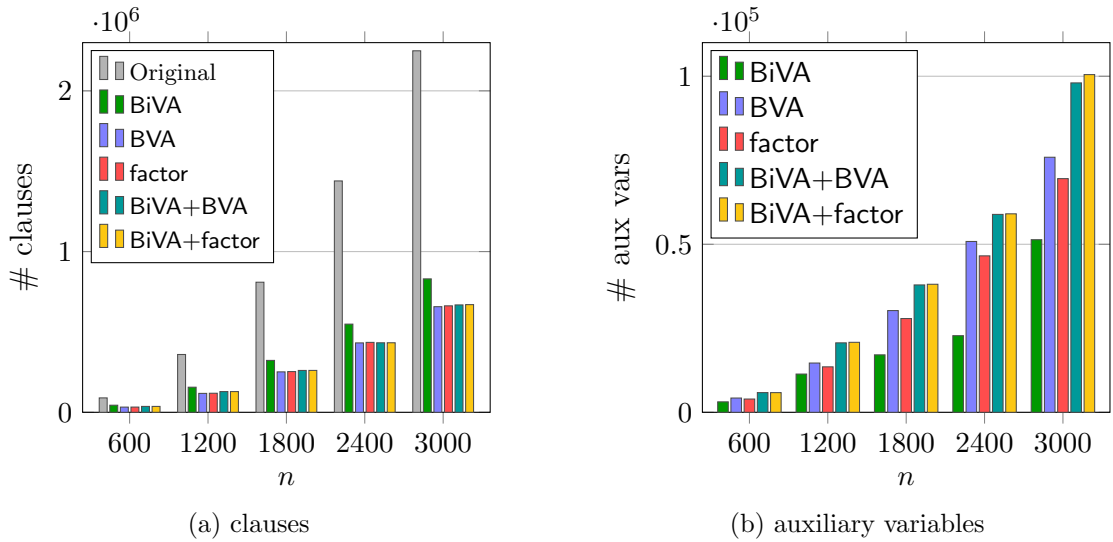
\begin{figure}
    \begin{subfigure}[b]{0.49\linewidth}
      \centering
    \begin{tikzpicture}
\begin{axis}[
    width=7cm, 
    height=6.5cm,
    ybar=1pt, 
    bar width=3pt, 
    enlarge x limits=0.15, 
    ymin=0,
    ymax=2300000,
    ylabel={\# clauses},
    xlabel={$n$},
    ylabel style={font=\large},
    yticklabel style={
        /pgf/number format/fixed,
        /pgf/number format/precision=1
    },
    symbolic x coords={600, 1200, 1800, 2400, 3000},
    xtick=data,
    ymajorgrids,
    tick label style={font=\normalsize},
    legend style={
        at={(0.02,0.98)}, 
        anchor=north west,
        font=\footnotesize,
        cells={anchor=west}
    }
]

\addplot[fill=gray!60!white, draw=black!70] coordinates {
    (600,89671) (1200,359732) (1800,809413) (2400,1439527) (3000,2249189)
};
\addlegendentry{Original}

\addplot[fill=green!60!black, draw=black!70] coordinates {
    (600,44190) (1200,156200) (1800,322931) (2400,548671) (3000,830189)
};
\addlegendentry{\textsf{BiVA}}

\addplot[fill=blue!50!white, draw=black!70] coordinates {
    (600,32404) (1200,117768) (1800,251632) (2400,431841) (3000,657162)
};
\addlegendentry{\textsf{BVA}}

\addplot[fill=red!70!white, draw=black!70] coordinates {
    (600,32583) (1200,118606) (1800,253436) (2400,435083) (3000,662432)
};
\addlegendentry{\textsf{factor}}

\addplot[fill=cyan!60!black, draw=black!70] coordinates {
    (600,36719) (1200,128923) (1800,260504) (2400,432518) (3000,668154)
};
\addlegendentry{\textsf{BiVA}+\textsf{BVA}}

\addplot[fill=lipicsYellow, draw=black!70] coordinates {
    (600,36700) (1200,128976) (1800,260705) (2400,432747) (3000,669987)
};
\addlegendentry{\textsf{BiVA}+\textsf{factor}}

\end{axis}
\end{tikzpicture}
    \caption{clauses}
    \label{fig:clause-compare}
    \end{subfigure}%
    \begin{subfigure}[b]{0.49\linewidth}
      \centering
    \begin{tikzpicture}
\begin{axis}[
    width=7cm, 
    height=6.5cm,
    ybar=1pt, 
    bar width=4pt, 
    enlarge x limits=0.15, 
    ymin=0,
    ymax=110000,
    ylabel={\# aux vars},
    xlabel={$n$},
    ylabel style={font=\large},
    yticklabel style={
        /pgf/number format/fixed,
        /pgf/number format/precision=1
    },
    symbolic x coords={600, 1200, 1800, 2400, 3000},
    xtick=data,
    ymajorgrids,
    tick label style={font=\normalsize},
    legend style={
        at={(0.02,0.98)}, 
        anchor=north west,
        font=\normalsize,
        cells={anchor=west}
    },
]

\addplot[fill=green!60!black, draw=black!70] coordinates {
    (600,3113) (1200,11382) (1800,17099) (2400,22800) (3000,51359)
};
\addlegendentry{\textsf{BiVA}}

\addplot[fill=blue!50!white, draw=black!70] coordinates {
    (600,4247) (1200,14662) (1800,30274) (2400,50834) (3000,75895)
};
\addlegendentry{\textsf{BVA}}

\addplot[fill=red!70!white, draw=black!70] coordinates {
    (600,3951) (1200,13531) (1800,27906) (2400,46555) (3000,69535)
};
\addlegendentry{\textsf{factor}}

\addplot[fill=cyan!60!black, draw=black!70] coordinates {
    (600,5863) (1200,20670) (1800,37903) (2400,58890) (3000,98018)
};
\addlegendentry{\textsf{BiVA}+\textsf{BVA}}

\addplot[fill=lipicsYellow, draw=black!70] coordinates {
    (600,5859) (1200,20828) (1800,38100) (2400,59063) (3000,100489)
};
\addlegendentry{\textsf{BiVA}+\textsf{factor}}
\end{axis}
\end{tikzpicture}
    \caption{auxiliary variables}
    \label{fig:aux-compare}
    \end{subfigure}
    \caption{Comparison of reencoding methods on random monotone formulas}
    \label{fig:compare-clauses-vars}
\end{figure}

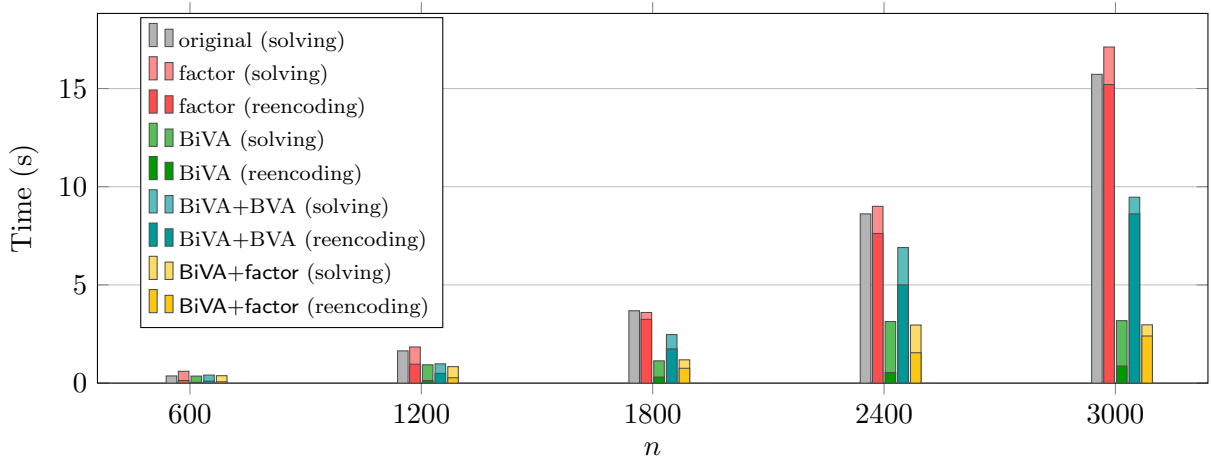
\begin{figure}
    \centering
   \begin{tikzpicture}
\begin{axis}[
    width=0.99\linewidth,
    height=6.5cm,
    ybar=1pt,
    bar width=4pt,
    xlabel={$n$},
    ymin=0,
    ylabel={Time (s)},
    symbolic x coords={600,1200,1800,2400,3000},
    xtick=data,
    enlarge x limits=0.10,
    y filter/.code={\pgfmathparse{#1/1000}\pgfmathresult},
    ymajorgrids,
    legend style={font=\scriptsize, at={(0.175,0.99)}, anchor=north, legend columns=1, cells={anchor=west}},
]
\addplot+[ybar,  bar shift=-7.00pt, draw=black!70, fill=gray!60!white] coordinates {(600,364.601) (1200,1641.021) (1800,3684.318) (2400,8617.097) (3000,15724.206)};
\addlegendentry{original (solving)}


\addplot+[ybar,  bar shift=-2.40pt, draw=black!70, fill=red!70!white!65] coordinates {(600,601.941) (1200,1840.980) (1800,3596.980) (2400,9003.040) (3000,17114.337)};
\addlegendentry{factor (solving)}
\addplot+[ybar, bar shift=-2.40pt, draw=black!70, fill=red!70!white] coordinates {(600,134.294) (1200,964.250) (1800,3244.749) (2400,7618.117) (3000,15205.774)};
\addlegendentry{factor (reencoding)}

\addplot+[ybar,  bar shift=2.40pt, draw=black!70, fill=green!60!black!65] coordinates {(600,357.450) (1200,931.889) (1800,1134.397) (2400,3134.673) (3000,3179.836)};
\addlegendentry{BiVA (solving)}
\addplot+[ybar, bar shift=2.40pt, draw=black!70, fill=green!60!black] coordinates {(600,37.626) (1200,132.192) (1800,306.632) (2400,535.168) (3000,867.490)};
\addlegendentry{BiVA (reencoding)}

\addplot+[ybar,  bar shift=7.20pt, draw=black!70, fill=cyan!60!black!65] coordinates {(600,407.571) (1200,979.925) (1800,2470.472) (2400,6903.736) (3000,9466.447)};
\addlegendentry{BiVA+BVA (solving)}
\addplot+[ybar, bar shift=7.20pt, draw=black!70, fill=cyan!60!black] coordinates {(600,99.392) (1200,495.173) (1800,1743.412) (2400,5002.723) (3000,8617.162)};
\addlegendentry{BiVA+BVA (reencoding)}

\addplot+[ybar,  bar shift=12.00pt,draw=black!70, fill=lipicsYellow!65] coordinates {(600,377.507) (1200,838.751) (1800,1185.270) (2400,2957.371) (3000,2968.519)};
\addlegendentry{\textsf{BiVA}+\textsf{factor} (solving)}
\addplot+[ybar, bar shift=12.00pt, draw=black!70, fill=lipicsYellow] coordinates {(600,66.519) (1200,272.411) (1800,754.325) (2400,1546.711) (3000,2398.409)};
\addlegendentry{\textsf{BiVA}+\textsf{factor} (reencoding)}
\end{axis}
\end{tikzpicture}
    \caption{End-to-end runtime comparison of reencoding methods on random monotone  formulas}
    \label{fig:runtime-both}
\end{figure}


\section*{Funding}
All authors were supported by NSF grant DMS-2434625. Przybocki was additionally supported by the NSF Graduate Research Fellowship Program under Grant
No. DGE-2140739. 

\bibliographystyle{abbrvurl}
\bibliography{references}

\appendix

\section{Proofs from \Cref{sec-bva-characterization}} \label{sec-sprn-appendix}
\begin{lemma}\label{lem:no-pos}
    Given a satisfiable 2-CNF formula $\varphi$, we can compute a formula $\varphi'$ in linear time such that (a) $\varphi'$ is obtained from $\varphi$ by replacing some variables by their negations wherever they appear, and (b) every clause in $\varphi'$ has at least one negative literal.
\end{lemma}
\begin{proof}
    Let $\varphi$ be a satisfiable 2-CNF formula over the variables $\vec{x}$, and let $\tau \colon \vec{x} \to \{\bot,\top\}$ be a satisfying assignment. Let $\varphi'$ be obtained from $\varphi$ by replacing each variable $x \in \vec{x}$ for which $\tau(x) = \top$ by its negation throughout the formula. Then, the assignment $x \mapsto \bot$ for each $x \in \vec{x}$ is a satisfying assignment of $\varphi'$, so every clause in $\varphi'$ has at least one negative literal. Since 2-SAT can be solved in linear time, this transformation can be carried out in linear time.
\end{proof}

\propsimpletogeneral*
\begin{proof}
    First, note that $g(n) = \Omega(n^2)$, because a simple 2-CNF formula can be of size $\Theta(n^2)$, and a reencoding algorithm must read its input. Thus, it suffices to show that every 2-CNF formula $\varphi$ on $n$ variables has a 2-CNF encoding with at most $f(n) + O(n)$ clauses, computable in time $g(n) + O(n^2)$.

    If $\varphi$ is unsatisfiable, which can be detected in time $O(n^2)$, then we use the trivial encoding $\{\{\}\}$ (an empty clause). Suppose from now on that $\varphi$ is satisfiable.

    First, using Tarjan's algorithm for strongly connected components~\cite{tarjan}, which runs in time $O(n^2)$, we can compute the partition $X_1 \sqcup X_2 \sqcup \dots \sqcup X_{n'}$ of the literals in $\varphi$ into strongly connected components. Given a literal $\ell \in X_i$, let $h(\ell)$ be a canonical representative from $X_i$. Now, let $\varphi^\star$ be the formula obtained from $\varphi$ by replacing each literal $\ell$ in $\varphi$ by $h(\ell)$. It suffices to construct a 2-CNF encoding $\varphi'^\star$ for $\varphi^\star$ with $|\varphi'^\star| \le f(n) + O(n)$, because then we can extend $\varphi'^\star$ to a 2-CNF encoding of $\varphi$ by adding the at most $4n$ clauses of the form $(\overline{\ell} \lor h(\ell))$ and $(\overline{h(\ell)} \lor \ell)$ for each literal $\ell$ in $\varphi$.

    Since the implication graph of $\varphi^\star$ contains no strongly connected components of size greater than $1$, it does not contain a pair of clauses of the form $(x_i \lor x_j)$ and $(\overline{x_i} \lor \overline{x_j})$, nor does it contain a pair of clauses of the form $(\overline{x_i} \lor x_j)$ and $(\overline{x_j} \lor x_i)$. If $\varphi^\star$ contains a pair of clauses of the form $(\ell_i \lor x_j)$ and $(\ell_i \lor \overline{x_j})$, then $\varphi^\star \models \ell_i$.
    Let $L$ be the set of literals $\ell_i$ such that $(\ell_i \lor x_j)$ and $(\ell_i \lor \overline{x_j})$ are both present in $\varphi^\star$ for some $x_j$. Then, since $L$ can be computed in time $O(n^2)$, we can compute
    \[
    \varphi^\dagger := \textsf{UnitPropagation}\left(\varphi^\star \land \bigwedge_{\ell \in L} \ell\right)
    \]
    in time $O(n^2)$. Now, $\varphi^\dagger \subseteq \varphi^\star$, so the implication graph of $\varphi^\dagger$ still contains no strongly connected components of size greater than $1$, and it contains no pair of clauses of the form $(\ell_i \lor x_j)$ and $(\ell_i \lor \overline{x_j})$. It suffices to construct a 2-CNF encoding $\varphi'^\dagger$ for $\varphi^\dagger$ with $|\varphi'^\dagger| \le f(n) + O(n)$, because then 
    \[\varphi'^\dagger \land \bigwedge_{\ell \in L} ((\ell \lor x_1) \land (\ell \lor \overline{x_1})),\]
    where $x_1$ is an arbitrary variable, is a 2-CNF encoding of $\varphi^\star$.

    Next, by \Cref{lem:no-pos}, we can compute in time $O(n^2)$ a formula $\varphi^\diamond$ that is obtained from $\varphi^\dagger$ by replacing a set $X \subseteq \textsf{Var}(\varphi)$ of variables by their negations such that every clause in $\varphi^\diamond$ has at least one negative literal. Then, $\varphi^\diamond$ is simple, and has at most $n$ variables, so by assumption we can compute an encoding $\varphi'^\diamond$ with $|\varphi'^\diamond| \leq f(n)$ in time $g(n)$.
    Finally, by inverting back in $\varphi'^\diamond$ the polarity of variables in $X$, we get an encoding for $\varphi^\dagger$ of the same size.
\end{proof}

The following definitions, which capture BVA steps in terms of an operation on SPRNs, are key to the proof of \Cref{thm:characterization}.

\begin{definition}[Coherent Biclique]
    If $\mathcal{R} = (B,A,E,p)$ is a PRN, then a \emph{coherent biclique} of $\mathcal{R}$ is a pair $(X,Y)$, where $X,Y \subset B \sqcup A$ are disjoint and nonempty, such that:
    \begin{enumerate}
        \item for every $x \in X$ and $y \in Y$, we have $(x, y) \in E$; and
        \item for every $y \in Y \cap B$ and $x_1, x_2 \in X$, 
        we have $p(x_1, y) = p(x_2, y)$.
    \end{enumerate}
\end{definition}

\begin{definition}[Biclique Reduction]\label{def:bic-red}
    Let $\mathcal{R} = (B,A,E,p)$ be a PRN and $(X,Y)$ a coherent biclique with edge set $E_{X,Y}$. Then, the \emph{biclique reduction} of $\mathcal{R}$ with respect to $(X,Y)$ is a new PRN $\mathcal{R}' = (B,A',E',p')$, where $A' = A \sqcup \{z\}$ for some new vertex $z$ and
    \begin{align*}
        E' &:= (E \setminus E_{X,Y}) \cup \{(x, z) \mid x \in X\} \cup \{(z, y) \mid y \in Y\}, \\
    p'(u, v) &:= \begin{cases}
        p(x, v) & \text{if } u = z \ \text{and} \ v \in Y \cap B, \ \text{where} \ x\in X \ \text{is arbitrary} \\
        p(u, v) & \text{otherwise}.
    \end{cases}
    \end{align*}
    We write $\mathcal{R} \hookrightarrow \mathcal{R}'$ if $\mathcal{R}'$ is the result of applying a biclique reduction to $\mathcal{R}$, and we write $\hookrightarrow^\ast$ for the reflexive transitive closure of $\hookrightarrow$.
\end{definition}
\noindent
The following theorem will be useful on the way to \Cref{thm:characterization}:
\begin{theorem}\label{thm:sequential-rectifier}
    Let $(G,p)$ be a polarized diagram. Then, we have $\mathcal{R}_{(G,p)} \hookrightarrow^\ast \mathcal{R}$ if and only if $\mathcal{R}$ is an SPRN realizing $(G,p)$.
\end{theorem}
\noindent
We prove each direction of~\Cref{thm:sequential-rectifier} individually.
\begin{lemma} \label{lem-reduction-sprn}
    Let $(G,p)$ be a polarized diagram. If $\mathcal{R}_{(G,p)} \hookrightarrow^\ast \mathcal{R}$, then $\mathcal{R}$ is an SPRN realizing $(G,p)$.

\end{lemma}
\begin{proof}
    It is direct from the definition that $\mathcal{R}_{(G,p)}$ is an SPRN realizing $(G,p)$. Thus, by induction, it suffices to show that if $\mathcal{R}_1$ is an SPRN realizing $(G,p)$ and $\mathcal{R}_1 \hookrightarrow \mathcal{R}_2$, then $\mathcal{R}_2$ is an SPRN realizing $(G,p)$. Assume $\mathcal{R}_1 = (B,A,E,p')$ is an SPRN realizing $(G,p)$. Let $\mathcal{D}(\mathcal{R}_1)$ be the set of all valid walks in $(G,p)$. Then, let $(X,Y)$ be the coherent biclique such that $\mathcal{R}_2 = (B,A',E',p'')$ is the result of reducing $(X,Y)$ in $\mathcal{R}_1$, and let $z$ be the unique vertex in $A' \setminus A$. Now, given any valid walk $\pi = (\pi_1, \pi_2, \dots, \pi_{k-1}, \pi_k)$ in $\mathcal{D}(\mathcal{R}_1)$, we define its image under a function $\varphi$ as the valid walk $\pi'$ in $\mathcal{R}_2$ obtained by replacing every edge $(\pi_i, \pi_{i+1})$ with $\pi_i \in X, \pi_{i+1} \in Y$ by the walk $(\pi_i, z, \pi_{i+1})$. Then, note that $\varphi$ is a bijection from $\mathcal{D}(\mathcal{R}_1)$ to $\mathcal{D}(\mathcal{R}_2)$ that preserves the endpoints; furthermore, by~\Cref{def:bic-red}, $p'(\pi) = p''(\varphi(\pi))$. It follows that $\mathcal{R}_2$ realizes $(G,p)$ and that for every $\{u,v\} \in \binom{B}{2}$, $\mathcal{R}_2$ contains at most one valid walk from $u$ to $v$ or from $v$ to $u$.
    
    To show that $\mathcal{R}_2$ is an SPRN, it remains to show that $z$ belongs to some valid walk. Let $x \in X$ and $y \in Y$ be vertices of $\mathcal{R}_1$. It suffices to show that there is a valid walk in $\mathcal{R}_1$ containing $(x,y)$, because then its image under $\varphi$ will contain $z$. Suppose first that $x,y \in B$. Then, $(x,y)$ is a valid walk in $\mathcal{R}_1$. Suppose next that $x \in B$ and $y \in A$. Then, $\mathcal{R}_1$ contains a valid walk $\pi' = (\pi'_1, \dots, y, \pi'_\ell, \dots, \pi'_k)$, from where $\pi'' := (x, y, \pi'_\ell, \dots, \pi'_k)$ is a valid walk. If $x \in A$ and $y \in B$, then $\mathcal{R}_1$ contains a valid walk $\pi' = (\pi'_1, \dots, x, \dots, \pi'_k)$, from where $\pi'' := (\pi'_1, \dots, x, y)$ is a valid walk. Finally, suppose $x,y \in A$. Then, there is a valid walk $\pi = (\pi_1, \pi_2, \dots, \pi_{k-1}, \pi_k)$ such that $\pi_i = x$ for some $i \in [2,k-1]$, and there is a valid walk $\pi' = (\pi'_1, \pi'_2, \dots, \pi'_{k'-1}, \pi'_{k'})$ such that $\pi'_j = y$ for some $j \in [2,k'-1]$. Then, $\pi'' = (\pi_1,\dots,\pi_i,\pi'_j,\dots,\pi'_{k'})$ is a valid walk in $\mathcal{R}_1$, and note that $(\pi_i,\pi'_j) = (x,y)$.
\end{proof}

\begin{lemma} \label{lem-sprn-reduction}
    Let $(G,p)$ be polarized diagram and $\mathcal{R}$ be an SPRN realizing $(G,p)$. Then, $\mathcal{R}_{(G,p)} \hookrightarrow^\ast \mathcal{R}$.
\end{lemma}
\begin{proof}
Write $\mathcal{R} = (B,A,E,p')$. The proof is by induction on $k := |A|$. For the base case, $k = 0$, observe that $\mathcal{R}_{(G,p)} = (V(G),\emptyset,E(G),p)$ is the only SPRN realizing $(G,p)$ with no auxiliary vertices. For the inductive case, suppose $\mathcal{R}$ is an SPRN realizing $(G,p)$ with $|A| = k+1$. Then,
let $a \in A$ be arbitrary, and define
\begin{align*}
    X &:= \{ x \in B \sqcup A \mid (x, a) \in E\} \\
    Y &:= \{ y \in B \sqcup A \mid (a, y) \in E\}.
\end{align*}
Let us denote by $E_a$  the set of edges of $\mathcal{R}$ incident to $a$, both ingoing and outgoing, and define the PRN $\mathcal{R}' = (B,A',E',p'')$ by:
\begin{align*}
    A' &:= A \setminus \{a\} \\
    E' &:= (E \setminus E_a) \cup \{ (x, y) \mid x \in X, y \in Y\} \\
    p''(u,v) &:= \begin{cases}
        p'(a,v) & \text{if } u \in X \ \text{and} \ v \in Y \cap B \\
        p'(u,v) & \text{otherwise}.
    \end{cases}
\end{align*}
We claim that $\mathcal{R}'$ is also an SPRN realizing $(G,p)$. We have $\mathcal{R}' \hookrightarrow \mathcal{R}$, since $\mathcal{R}$ is precisely the result of reducing the coherent biclique $(X,Y)$ in $\mathcal{R}'$. As in the proof of \Cref{lem-reduction-sprn}, there is a bijection $\varphi$ from the valid walks of $\mathcal{R}'$ to the valid walks of $\mathcal{R}$ that preserves the endpoints and polarities of the paths. It follows that $\mathcal{R}'$ realizes $(G,p)$ and that for every $\{u,v\} \in \binom{B}{2}$, $\mathcal{R}'$ contains at most one valid walk from $u$ to $v$ or from $v$ to $u$. Also, given any $z \in A' \subset A$, there is a valid walk $\pi$ in $\mathcal{R}$ containing $z$. Then, $\varphi^{-1}(\pi)$ is a valid walk in $\mathcal{R}'$ containing $z$. Therefore, $\mathcal{R}'$ is an SPRN.

We have $|A'| = k$, so by the inductive hypothesis, $\mathcal{R}_{(G,p)} \hookrightarrow^\ast \mathcal{R}'$. But also, $\mathcal{R}' \hookrightarrow \mathcal{R}$, so $\mathcal{R}_{(G,p)} \hookrightarrow^\ast \mathcal{R}$, as desired.
\end{proof}

\begin{proof}[Proof of \Cref{thm:sequential-rectifier}]
    This is immediate from \Cref{lem-sprn-reduction,lem-reduction-sprn}.
\end{proof}
\noindent
Next, we prove the forward direction of \Cref{thm:characterization}.
\begin{lemma} \label{lem-characterization-forward}
    Let $\varphi$ and $\varphi'$ be 2-CNF formulas, where $\varphi$ is simple. If $\varphi \bvachain \varphi'$, then there is an SPRN $\mathcal{R}$ realizing $G_{\varphi}$ with $F_\mathcal{R} \cong_W \varphi'$, where $W = \textsf{Var}(\varphi') \setminus \textsf{Var}(\varphi)$.
\end{lemma}
\begin{proof}
    First, we show that there is some Horn 2-CNF formula $\varphi^*$ (i.e., each clause contains a negative literal) such that $\varphi^* \cong_W \varphi'$. Let $\tau \colon \textsf{Var}(\varphi) \to \{\bot,\top\}$ be the assignment $x \mapsto \bot$ for all $x \in \textsf{Var}(\varphi)$. Since $\varphi$ is simple, and hence Horn, $\tau$ is a satisfying assignment, so there is an extension $\tau' \colon \textsf{Var}(\varphi') \to \{\bot,\top\}$ satisfying $\varphi'$. Let $\varphi^*$ be obtained from $\varphi$ by replacing each variable $x \in W$ for which $\tau'(x) = \top$ by its negation throughout the formula. Then, the assignment $x \mapsto \bot$ for each $x \in \textsf{Var}(\varphi^*)$ is a satisfying assignment of $\varphi^*$, so every clause of $\varphi^*$ has at least one negative literal, and thus $\varphi^*$ is Horn.

    We define a polarized diagram $(G,p)$ of $G_{\varphi}$ as follows. We must specify, for each $\{x_i,x_j\} \in G_{\varphi}$ (which corresponds to a clause $(\overline{x_i} \lor \overline{x_j}) \in \varphi$) how to orient this edge in $G$. If $(\overline{x_i} \lor \overline{x_j}) \in \varphi^*$, then orient $\{x_i,x_j\}$ arbitrarily in $G$. Otherwise, since $\varphi^*$ is an encoding of $\varphi$ by \Cref{lemma:bva-encoding} and is 2-CNF, $\varphi^*$ either contains a set of clauses of the form
    \[
        (\overline{x_i} \lor y_1), (\overline{y_1} \lor y_2), \dots, (\overline{y_{k-1}} \lor y_k), (\overline{y_k} \lor \overline{x_j}) \quad (\text{i.e., } x_i \rightarrow y_1 \rightarrow y_2 \rightarrow \dots \rightarrow \overline{x_j})
    \]
    or of the form
    \[
        (\overline{x_j} \lor y_1), (\overline{y_1} \lor y_2), \dots, (\overline{y_{k-1}} \lor y_k), (\overline{y_k} \lor \overline{x_i}) 
       \quad (\text{i.e., } x_j \rightarrow y_1 \rightarrow y_2 \rightarrow \dots \rightarrow \overline{x_i}),
    \]
    where $y_1,\dots,y_k$ are auxiliary variables. If the former case holds, then orient $\{x_i,x_j\}$ as $(x_i,x_j)$ in $G_{\varphi}$; otherwise, orient $\{x_i,x_j\}$ as $(x_j,x_i)$ in $G_{\varphi}$.

    Now, we prove that there is an SPRN $\mathcal{R}$ realizing $(G,p)$ with $F_\mathcal{R} = \varphi^*$. The proof is by induction on the number of BVA steps in $\varphi \bvachain \varphi^*$. For the base case, of $0$ steps, we have $\varphi^* = \varphi$, and $\mathcal{R}_{(G,p)}$ is an SPRN realizing $(G,p)$ with $F_{\mathcal{R}_{(G,p)}} = \varphi$. For the inductive step, suppose that there is an SPRN $\mathcal{R}'$ realizing $(G,p)$, let $\varphi'' := F_{\mathcal{R}'}$, and suppose that $\varphi'' \bvastep \varphi^*$. Our goal is to show that there is an SPRN $\mathcal{R}$ realizing $(G,p)$ with $F_{\mathcal{R}} = \varphi^*$. Write $\mathcal{R}' = (B,A, E', p')$ and $\mathcal{R} = (B,A \sqcup \{z\},E,p)$. Suppose that the BVA step applied to $\varphi''$ consists of replacing the clauses $\mathcal{C} \bowtie \mathcal{D}$ with $\{(C_i \lor z) \mid C_i \in \mathcal{C}\} \cup \{(\overline{z} \lor D_j) \mid D_j \in \mathcal{D}\}$. Since $\varphi^*$ is a Horn 2-CNF formula, we have that $\mathcal{C}$ is a set of negative literals, and $\mathcal{D}$ is a set of literals. Let $X$ and $Y$ be the sets of variables occurring in $\mathcal{C}$ and $\mathcal{D}$ respectively (so $X = \{x \mid \overline{x} \in \mathcal{C}\}$). Then, $(X,Y)$ is a coherent biclique in $\mathcal{R}'$. Let $\mathcal{R}$ be the result of applying a biclique reduction of $\mathcal{R}'$ with respect to $(X,Y)$, and let $z$ be the newly introduced auxiliary vertex. Then, $\mathcal{R}$ is the result of replacing the edges from $(X,Y)$ with the edges $\{(x, z) \mid x \in X\} \cup \{(z, y) \mid y \in Y\}$, where, for all $y \in Y \cap B$, $p(z,y) = p'(x,y)$ for an arbitrary $x \in X$. Thus, in $F_{\mathcal{R}}$, the clauses $\mathcal{C} \bowtie \mathcal{D}$ are replaced by the clauses $(\overline{x} \lor z)$ for every $x \in X$, $(\overline{z} \lor \overline{y})$ for every $y \in Y \cap B$ with $p(z,y) = -$, and $(\overline{z} \lor y)$ for every other $y \in Y$. That is, 
    \[
    F_{\mathcal{R}} = (F_{\mathcal{R}'} \setminus (\mathcal{C} \bowtie \mathcal{D})) \cup (\{(C_i \lor z) \mid C_i \in \mathcal{C}\} \cup \{(\overline{z} \lor D_j) \mid D_j \in \mathcal{D}\}).
    \]
   Thus, $F_\mathcal{R} = \varphi^*$. Since $\mathcal{R}'$ is an SPRN realizing $(G, p)$, by the backward direction of \Cref{thm:sequential-rectifier} we have $\mathcal{R}_{(G,p)} \hookrightarrow^\ast \mathcal{R}'$. Since $\mathcal{R}' \hookrightarrow \mathcal{R}$, we have $\mathcal{R}_{(G,p)} \hookrightarrow^\ast \mathcal{R}$. Hence, by the forward direction of \Cref{thm:sequential-rectifier}, $\mathcal{R}$ is an SPRN realizing $(G,p)$.
\end{proof}
\noindent
For the backward direction of \Cref{thm:characterization}, we need a few lemmas.
\begin{lemma} \label{lem-commute}
    If $\varphi \bvachain \varphi'$ and $\varphi' \cong_W \varphi''$, where $W = \textsf{Var}(\varphi') \setminus \textsf{Var}(\varphi)$, then $\varphi \bvachain \varphi''$.
\end{lemma}
\begin{proof}
    We prove the lemma by induction on the number of BVA steps in $\varphi \bvachain \varphi'$. In the base case of 0 steps, we have $\varphi' = \varphi$, and the lemma is trivial. For the inductive step, suppose that $\varphi \bvachain \varphi^\dagger$ and for all $\varphi^\star \cong_{W'} \varphi^\dagger$, where $W' = \textsf{Var}(\varphi^\dagger) \setminus \textsf{Var}(\varphi)$, we have $\varphi \bvachain \varphi^\star$. Let $\varphi'$ be such that  $\varphi^\dagger \bvastep \varphi'$ and $\varphi' \cong_W \varphi''$, where $W = \textsf{Var}(\varphi') \setminus \textsf{Var}(\varphi)$. Our goal is to show that $\varphi \bvachain \varphi''$. Let $\Delta \subseteq W$ be the smallest subset such that $\varphi' \cong_\Delta \varphi''$, and let $z$ be the unique element of $\textsf{Var}(\varphi') \setminus \textsf{Var}(\varphi^\dagger)$. Let $\varphi^\star$ be the formula resulting from $\varphi^\dagger$ by replacing every variable from $\Delta \setminus \{z\}$ by its negation wherever it appears. By the inductive hypothesis, $\varphi \bvachain \varphi^\star$, so it suffices to show that $\varphi^\star \bvastep \varphi''$.

    Let $\mathcal{C}$ and $\mathcal{D}$ be such that 
    \[ \varphi' = \left(\varphi \setminus (\mathcal{C} \bowtie \mathcal{D})\right) \cup
        \{(C_i \lor z) \mid C_i \in \mathcal{C}\} \cup \{(\overline{z} \lor D_j) \mid D_j \in \mathcal{D}\}.
    \]
   Now, let $\mathcal{C}'$ (respectively, $\mathcal{D}'$) be the result of replacing every variable from $\Delta \setminus \{z\}$ in $\mathcal{C}$ (respectively, $\mathcal{D}$) by its negation. Then, $\mathcal{C}' \bowtie \mathcal{D}' \subseteq \varphi^\star$, so we have $\varphi^\star \bvastep \varphi_1$ and $\varphi^\star \bvastep \varphi_2$, where
   \begin{align*}
       \varphi_1 &= \left(\varphi^\star \setminus (\mathcal{C}' \bowtie \mathcal{D}')\right) \cup
        \{(C'_i \lor z) \mid C'_i \in \mathcal{C}'\} \cup \{(\overline{z} \lor D'_j) \mid D'_j \in \mathcal{D'}\} \\
        \varphi_2 &= \left(\varphi^\star \setminus (\mathcal{D}' \bowtie \mathcal{C}')\right) \cup
        \{(D'_i \lor z) \mid D'_i \in \mathcal{D}'\} \cup \{(\overline{z} \lor C'_j) \mid C'_j \in \mathcal{C'}\}.
   \end{align*}
    Notice that $\varphi_1$ results from $\varphi'$ by replacing every variable from $\Delta \setminus \{z\}$ by its negation wherever it appears. If $z \notin \Delta$, then $\varphi_1 = \varphi''$. Otherwise, since $\varphi_2$ is the result of replacing every instance of $z$ by its negation in $\varphi_1$, we have $\varphi_2 = \varphi''$. In either case, $\varphi^\star \bvastep \varphi''$, as desired.
\end{proof}

\begin{lemma} \label{lem-biclique-implies-bva}
    Let $\mathcal{R}_1$ and $\mathcal{R}_2$ be PRNs such that $\mathcal{R}_1 \hookrightarrow \mathcal{R}_2$. Then, $F_{\mathcal{R}_1} \bvastep F_{\mathcal{R}_2}$.
\end{lemma}
\begin{proof}
    Write $\mathcal{R}_1 = (B,A,E,p)$ and $\mathcal{R}_2 = (B,A \sqcup \{z\}, E', p')$. Suppose that $\mathcal{R}_1 \hookrightarrow \mathcal{R}_2$. Then, $\mathcal{R}_2$ is a biclique reduction of $\mathcal{R}_1$ with respect to some coherent biclique $(X,Y)$. Thus, $\mathcal{R}_1$ has edges $(x,y)$ for every $x \in X$ and $y \in Y$, and for every $y \in Y \cap B$ and $x_1, x_2 \in X$, we have $p(x_1,y) = p(x_2,y)$.   
    For $y \in Y$, we define the literal 
    \[
    \rho(y) :=  \begin{cases}
        \overline{y} & \text{if } y \in B \text{ and } p(x, y) = -\\
        y & \text{otherwise}.
    \end{cases}
    \]
   Note that for every $x \in X$ and $y \in Y$, $F_{\mathcal{R}_1}$ contains the clause $(\overline{x} \lor \rho(y))$. That is, $F_{\mathcal{R}_1}$ contains the clauses $\mathcal{C} \bowtie \mathcal{D}$, where $\mathcal{C} = \{\overline{x} \mid x \in X\}$ and $\mathcal{D} = \{ \rho(y) \mid y \in Y\}$.
    
    In $\mathcal{R}_2$, the edges from $(X, Y)$ are replaced by the edges $\{(x, z) \mid x \in X\} \cup \{(z, y) \mid y \in Y\}$, where, for all $y \in Y \cap B$, $p'(z,y) = p(x,y)$ for an arbitrary $x \in X$. Thus, in $F_{\mathcal{R}_2}$, the clauses $\mathcal{C} \bowtie \mathcal{D}$ are replaced by the clauses $(\overline{x} \lor z)$ for every $x \in X$, $(\overline{z} \lor \overline{y})$ for every $y \in Y \cap B$ with $p'(z,y) = -$, and $(\overline{z} \lor y)$ for every other $y \in Y$. That is, $F_{\mathcal{R}_2}$ contains the clauses $\{(C_i \lor z) \mid C_i \in \mathcal{C}\} \cup \{(\overline{z} \lor D_j) \mid D_j \in \mathcal{D}\}$. Thus, $F_{\mathcal{R}_1} \bvastep F_{\mathcal{R}_2}$.
\end{proof}
\noindent
Now we can prove the backward direction of \Cref{thm:characterization}.
\begin{lemma} \label{lem-characterization-backward}
    Let $\varphi$ and $\varphi'$ be 2-CNF formulas, where $\varphi$ is simple. If there is an SPRN $\mathcal{R}$ realizing $G_{\varphi}$ with $F_\mathcal{R} \cong_W \varphi'$, where $W = \textsf{Var}(\varphi') \setminus \textsf{Var}(\varphi)$, then we have $\varphi \bvachain \varphi'$.
\end{lemma}
\begin{proof}
    Suppose that $\mathcal{R}$ is an SPRN realizing $G_\varphi$ with $F_\mathcal{R} \cong_W \varphi'$. Then, by \Cref{thm:sequential-rectifier}, there is a polarized diagram $(G,p)$ of $G_\varphi$ such that $\mathcal{R}_{(G,p)} \hookrightarrow^\ast \mathcal{R}$. Then, by \Cref{lem-biclique-implies-bva}, $F_{\mathcal{R}_{(G,p)}} \bvachain F_{\mathcal{R}}$, and since $\varphi = F_{\mathcal{R}_{(G,p)}}$ (\Cref{remark:varphi-eq-F}), we have $\varphi \bvachain F_{\mathcal{R}}$. Then, as $\varphi' \cong_W F_\mathcal{R}$, we have $\varphi \bvachain \varphi'$ by \Cref{lem-commute}.
\end{proof}

\begin{proof}[Proof of \Cref{thm:characterization}]
    This is immediate from \Cref{lem-characterization-forward,lem-characterization-backward}.
\end{proof}

\section{Proofs from \Cref{sec-2cnf-general}}

\thmnechiporukdiagram*
\begin{proof}
    Let $G_\varphi$ be the associated diagram of $\varphi$. By \Cref{thm:characterization}, it suffices to build an SPRN $\mathcal{R}$ realizing $G_\varphi$ with at most $\left(\frac{\lg(3)}{4} + o(1)\right) \frac{n^2}{\lg n}$ edges. Since $\varphi$ is simple, the directed part of $G_\varphi$ is acyclic, so we can topologically order the vertices $v_1, \dots, v_n$ so that every directed edge $(v_i,v_j)$ satisfies $i < j$. We create a polarized diagram $(G,p)$ of $G_\varphi$ by orienting each $\{v_i,v_j\} \in E(G_\varphi)$ to be $(v_i,v_j)$, where $i < j$. Our SPRN $\mathcal{R} = (V(G), A, E, p')$ will be a realization of $(G,p)$ and thus of $G_\varphi$.
    \subparagraph*{Construction.}
    Let $\log_3 x = \frac{\lg x}{\lg 3}$. Let $q = \lfloor n / \log_3^2 n\rfloor$, and partition $V(G)$ into blocks $B_1, \ldots, B_{\lceil n/q\rceil}$ of size at most $q$; specifically, let $B_i = \{v_{(i-1)\cdot q + 1}, \dots, v_{i\cdot q}\}$ or, if $i \cdot q > n$, a truncation thereof. Within each block $B_i$, for each pair of vertices $\{u, v\} \in \binom{B_i}{2}$ create in $\mathcal{R}$ an auxiliary vertex  $z_{\{u, v\}}$, and create edges $(u,z_{\{u, v\}})$ and $(v,z_{\{u, v\}})$. Now, let $r = \lfloor \log_3 n  - 3\log_3 \log_3 n\rfloor$, and partition as well $V(G)$ into parts  $P_1, \ldots, P_{\lceil n/r\rceil}$ of size at most $r$; specifically, let $P_i = \{v_{(i-1)\cdot r + 1}, \dots, v_{i\cdot r}\}$ or, if $i \cdot r > n$, a truncation thereof. For each part $P_i$ and edge $(u, v) \in E(G) \cap (P_i \times P_i)$, create in $\mathcal{R}$ a directed edge $(u,v)$ with $p'(u,v) = p(u,v)$.

    Now, for each part $P_i$ and each function $S \colon P_i \to \{-,0,+\}$, create in $\mathcal{R}$ an auxiliary vertex $y_S$, and create a directed edge $(y_S,v)$ with $p'(y_S,v) = -$ for every $v \in P_i$ satisfying $S(v) = -$, and create a directed edge $(y_S,v)$ with $p'(y_S,v) = +$ for every $v \in P_i$ satisfying $S(v) = +$. Given a vertex $v \in V(G)$, let $N^\pm(v) = \{w \in V(G) \mid (v,w) \in E(G) \ \text{and} \ p(v,w) = \pm\}$. Then, for each $S \colon P_i \to \{-,0,+\}$, let 
    \[A(S) := \{ v \in P_j \mid j < i, \ N^-(v) \cap P_i = S^{-1}(-), \ \text{and} \ N^+(v) \cap P_i = S^{-1}(+)\}.
    \]
    Finally, for each $\ell \in [ \lceil n/q \rceil]$, let $A_\ell(S) := A(S) \cap B_\ell$, and let $a_1, \ldots, a_k$ be the elements of $A_\ell(S)$ in an arbitrary order. Then, create in $\mathcal{R}$ directed edges $(z_{\{a_{2j-1}, a_{2j}\}}, y_S)$ for $j \in [\lfloor k/2\rfloor]$, and if $k$ is odd, create the directed edge $(a_k,y_S)$. If $\mathcal{R}$ contains any auxiliary vertices not belonging to a valid walk, then delete them and their incident edges.

    \subparagraph*{Correctness.} Now, we claim that $\mathcal{R}$ is a strict rectifier network for $(G,p)$. Suppose $(u,v) \in E(G)$. If $u,v \in P_i$, then $\mathcal{R}$ contains the edge $(u,v)$ with $p'(u,v) = p(u,v)$. Otherwise, let $i,j \in [\lceil n/r\rceil]$ be such that $u \in P_i$ and $v \in P_j$; then, $i < j$ by the topological ordering. Let $S \colon P_j \to \{-,0,+\}$ be such that $u \in A(S)$; such an $S$ is uniquely defined by the conditions $S^{-1}(-) = N^-(u) \cap P_j$ and $S^{-1}(+) = N^+(u) \cap P_j$. Then, either there is a unique $u'$ such that $(z_{\{u,u'\}}, y_S)$ is a directed edge, in which case we have the valid walk $(u,z_{\{u,u'\}},y_S,v)$, or we have the directed edge $(u,y_S)$, in which case we have the valid walk $(u,y_S,v)$. Also, $p'(y_S,v) = p(u,v)$ by our choice of $S$. Thus, if $(u,v) \in E(G)$, then there is a valid walk $\pi$ from $u$ to $v$ with $p'(\pi) = p(u,v)$ in $\mathcal{R}$.

    Conversely, suppose that $\mathcal{R}$ contains a valid walk $\pi$ from $u$ to $v$ for some $u,v \in B$. Again, let $i,j \in [\lceil n/r\rceil]$ be such that $u \in P_i$ and $v \in P_j$. If $i = j$, then the only possibility for the valid walk is $(u,v)$, which only happens if $(u,v) \in E(G)$ and $p(u,v) = p'(u,v) = p'(\pi)$. Otherwise, $i < j$, and the valid walk is either of the form $(u,z_{\{u,u'\}},y_S,v)$ or of the form $(u,y_S,v)$ for some $u' \in B \setminus \{u\}$ and $S \colon P_j \to \{-,0,+\}$; furthermore, it is impossible to have both valid walks $(u,z_{\{u,u'\}},y_S,v)$ and $(u,y_S,v)$. We must have $u \in A(S)$, so $(u,v) \in E(G)$ with $p(u,v) = S(v) = p'(y_S,v) = p'(\pi)$. Thus, we have $(u,v) \in E(G)$ if and only if there is a valid walk $\pi$ from $u$ to $v$ in $\mathcal{R}$ such that $p'(\pi) = p(u, v)$, in which case $\pi$ is the unique valid walk between $u$ and $v$. Therefore, $\mathcal{R}$ is an SPRN realizing $(G,p)$.

    \subparagraph*{Counting the edges.}
    We now show that $\mathcal{R}$ has at most $\left(\frac{\lg(3)}{4} + o(1)\right) \frac{n^2}{\lg n}$ edges. The number of edges of the form $(u,z_{\{u, v\}})$ is at most
    \[
    \lceil n/q\rceil \cdot \binom{q}{2} \leq \left(\frac{n}{q}+1\right) \cdot \frac{q^2}{2} = \frac{nq}{2} + \frac{q^2}{2} \leq \frac{n^2}{2\log_3^2 n} + \frac{n^2}{2\log_3^4 n} = o(n^2/\lg n).
    \]
    The number of edges of the form $(u,v)$ for some $i \in [\lceil n/r\rceil]$ and $u,v \in P_i$ is at most
    \[
        \lceil n/r\rceil \cdot \binom{r}{2} \le \frac{nr}{2} + \frac{r^2}{2} \le \frac{n \log_3 n + \log_3^2 n}{2} = o(n^2/\lg n).
    \]
    The number of edges of the form $(y_S,v)$ for some $i \in [\lceil n/r\rceil]$, $S \colon P_i \to \{-,0,+\}$, and $v \in S^{-1}(\{-,+\})$ is at most
    \[
        \lceil n/r\rceil \cdot 3^r \cdot r \le \left(\frac{n}{r}+1\right) \cdot \frac{n}{\log_3^3 n} \cdot r = \frac{n^2}{\log_3^3 n} + \frac{nr}{\log_3^3 n} = o(n^2 / \lg n).
    \]
    
    We now count the number of edges of the form $(z_{\{a_{2j-1}, a_{2j}\}}, y_S)$ for some $i \in [\lceil n/r\rceil]$, $S \colon P_i \to \{-,0,+\}$, $a_{2j-1}, a_{2j} \in A_\ell(S)$, and $\ell \in [\lceil n/q \rceil]$. Fix $i \in [\lceil n/r\rceil]$ and some $S \colon P_i \to \{-,0,+\}$. Then, for every $\ell \in [\lceil n/q\rceil]$ there are $\left\lfloor |A_\ell(S)|/2\right\rfloor$ edges of the form $(z_{\{a_{2j-1}, a_{2j}\}}, y_S)$.
    Since $A(S) = \bigcup_{\ell \in [\lceil n/q\rceil]} A_\ell(S)$, for every $S$ we have
    \(\sum_{\ell \in [\lceil n/q\rceil]} \lfloor |A_\ell(S)|/2\rfloor \leq  |A(S)|/2.
    \)
    Now, note that $v$ can only belong to $A(S)$, for $S\colon P_i \to \{-,0,+\}$, if $v \in P_j$ for some $j < i$; furthermore, if $j < i$, then $v \in A(S)$ for a unique $S\colon P_i \to \{-,0,+\}$. Thus, by exchanging summations, we have
    \begin{align*}
     \frac{1}{2}\sum_{\substack{i \in [\lceil n/r\rceil] \\ S \colon P_i \to \{-,0,+\}}}  |A(S)| &= \frac{1}{2} \sum_{v \in V(G)}  \sum_{\substack{i \in [\lceil n/r\rceil] \\ S \colon P_i \to \{-,0,+\}}} \mathbbm{1}_{v \in A(S)}\\
    &=\frac{1}{2} \sum_{i \in [\lceil n/r\rceil]} \sum_{j < i} |P_j| = \frac{1}{2} \binom{\lceil n/r\rceil}{2} \cdot r = \left(\frac{1}{4} + o(1)\right) \frac{n^2}{\log_3 n}.
      \end{align*}
    Therefore, the number of edges of the form $(z_{\{a_{2j-1}, a_{2j}\}}, y_S)$ is at most $\left(\frac{\lg(3)}{4} + o(1)\right) \frac{n^2}{\lg n}$.

    Finally, the number of edges of the form $(a_k,y_S)$ for some $i \in [\lceil n/r\rceil]$, $S \colon P_i \to \{-,0,+\}$, $a_k \in A_\ell$, and $\ell \in [\lceil n/q \rceil]$ is at most
    \[
        \lceil n/r\rceil \cdot 3^r \cdot \lceil n/q \rceil \le \left(\frac{n}{r}+1\right) \cdot \frac{n}{\log_3^3 n} \cdot \left(\frac{n}{q} + 1\right) = o(n^2 / \lg n).
    \]
    Thus, the total number of edges in $\mathcal{R}$ is at most $\left(\frac{\lg(3)}{4} + o(1)\right) \frac{n^2}{\lg n}$.
\end{proof}

\lemcountformulas*
\begin{proof}
    In a 2-CNF formula with $m$ clauses, the number of variables is at most $2m$, so the number of 2-CNF formulas with $m$ clauses is at most $\binom{(4m)^2}{m}$. Hence, the number of 2-CNF formulas at most $m$ clauses is at most
    \[
        \binom{(4m)^2}{m+1} \le \left(\frac{(4m)^2 \cdot e}{m+1}\right)^{m+1} = m^{(1+o(1))m} = 2^{(1+o(1)) m \cdot \lg m}. \qedhere
    \]
\end{proof}

\lemcountsimple*
\begin{proof}
    We construct a simple 2-CNF formula $\varphi$ as follows. Let our variables be $x_1, \dots, x_n$. For every pair $\{x_i, x_j\}$ with $i < j$, do exactly one of the following:
    \begin{enumerate}
        \item add the clause $(\overline{x_i} \lor \overline{x_j})$ to $\varphi$;
        \item add the clause $(\overline{x_i} \lor x_j)$ to $\varphi$; or
        \item add neither clause to $\varphi$.
    \end{enumerate}
    There are $3^{\binom{n}{2}} = 3^{(1/2-o(1)) n^2}$ ways to carry out this construction, each of which yields a distinct simple 2-CNF formula on $n$ variables.
\end{proof}

\lowerbound*
\begin{proof}
    Let $g(m)$ be the number of 2-CNF formulas with at most $m$ clauses, and let $h(n) \coloneqq 2^{\binom{n}{2}}$ be the number of monotone 2-CNF functions on $n$ variables. By \Cref{lem-count-formulas}, $\lg g(m) \le (1+o(1)) m \cdot \lg m$. If every monotone 2-CNF function on $n$ variables can be encoded with at most $m$ clauses, then $g(m) \ge h(n)$, and consequently $\lg g(m) \geq \lg h(n)$. 
    If $m \le (r+o(1)) n^2/\lg n$ for some constant $r$, then $\lg m \le (2+o(1)) \lg n$, and thus
    \[
        (1/2 - o(1))n^2 \le \lg h(n) \le \lg g(m) \le (1+o(1)) m \cdot \lg m \le (2r+o(1)) n^2,
    \]
    from where $r \ge 1/4$; that is, $m \ge (\frac{1}{4} - o(1)) n^2/\lg n$.
\end{proof}

\idealizedbvamonotone*
\begin{proof}
    It suffices to prove the theorem for antitone 2-CNF formulas, which is more convenient because they are simple. So let $\varphi$ be an antitone 2-CNF formula, and let $G_\varphi$ be the associated diagram of $\varphi$. By \Cref{thm:characterization}, it suffices to build an SPRN $\mathcal{R}$ realizing $G_\varphi$ with at most $\left(\frac{1}{4} + o(1)\right) \frac{n^2}{\lg n}$ edges. Let $v_1, \dots, v_n$ be an arbitrary enumeration of $V(G_\varphi)$. We create a polarized diagram $(G,p)$ of $G_\varphi$ by orienting each $\{v_i,v_j\} \in E(G_\varphi)$ to be $(v_i,v_j)$, where $i < j$. Our SPRN $\mathcal{R} = (V(G), A, E, p')$ will be a realization of $(G,p)$ and thus of $G_\varphi$.

    Let $q = \lfloor n / \lg^2 n\rfloor$, and partition $V(G)$ into blocks $B_1, \ldots, B_{\lceil n/q\rceil}$ of size at most $q$; specifically, let $B_i = \{v_{(i-1)\cdot q + 1}, \dots, v_{i\cdot q}\}$ or, if $i \cdot q > n$, a truncation thereof. Within each block $B_i$, for each pair of vertices $\{u, v\} \in \binom{B_i}{2}$ create in $\mathcal{R}$ an auxiliary vertex  $z_{\{u, v\}}$, and create edges $(u,z_{\{u, v\}})$ and $(v,z_{\{u, v\}})$. Now, let $r = \lfloor \lg n  - 3\lg \lg n\rfloor$, and partition as well $V(G)$ into parts  $P_1, \ldots, P_{\lceil n/r\rceil}$ of size at most $r$; specifically, let $P_i = \{v_{(i-1)\cdot r + 1}, \dots, v_{i\cdot r}\}$ or, if $i \cdot r > n$, a truncation thereof. For each part $P_i$ and edge $(u, v) \in E(G) \cap (P_i \times P_i)$, create in $\mathcal{R}$ a directed edge $(u,v)$ with $p'(u,v) = p(u,v) = -$.

    Now, for each part $P_i$ and each subset $S \subseteq P_i$, create in $\mathcal{R}$ an auxiliary vertex $y_S$, and create a directed edge $(y_S,v)$ with $p'(y_S,v) = -$ for each $v \in S$. Then, for each $S \subseteq P_i$, let $A(S) := \{ v \in P_j \mid j < i \ \text{and} \ N_G(v) \cap P_i = S \}$. Finally, for each $\ell \in [ \lceil n/q \rceil]$, let $A_\ell(S) := A(S) \cap B_\ell$, and let $a_1, \ldots, a_k$ be the elements of $A_\ell(S)$ in an arbitrary order. Then, create in $\mathcal{R}$ directed edges $(z_{\{a_{2j-1}, a_{2j}\}}, y_S)$ for $j \in [\lfloor k/2\rfloor]$, and if $k$ is odd, create the directed edge $(a_k,y_S)$. If $\mathcal{R}$ contains any auxiliary vertices not belonging to a valid walk, then delete them and their incident edges.

    The proof that the construction is correct and the computation of the number of edges is very similar to the proof of \Cref{thm:nechiporuk-diagram}.
\end{proof}

\begin{theorem}[Formal version of \Cref{thm-idealized-bva-no-simp}] \label{thm:bva-no-preprocessing}
    For every 2-CNF formula $\varphi$ on $n$ variables, there is a 2-CNF formula $\varphi'$ such that $\varphi \bvachain \varphi'$ and $|\varphi'| \le \left(1 + o(1)\right) \frac{n^2}{\lg n}$.
\end{theorem}
\begin{proof}
    Let the variables of $\varphi$ be $x_1,\dots,x_n$. We can write $\varphi = \varphi_1 \land \varphi_2 \land \varphi_3 \land \varphi_4$, where
    \begin{align*}
        \varphi_1 &:= \{(x_i \lor x_j) \mid i < j \ \text{and} \ (x_i \lor x_j) \in \varphi\} \\
        \varphi_2 &:= \{(\overline{x_i} \lor x_j) \mid i < j \ \text{and} \ (\overline{x_i} \lor x_j) \in \varphi\} \\
        \varphi_3 &:= \{(x_i \lor \overline{x_j}) \mid i < j \ \text{and} \ (x_i \lor \overline{x_j}) \in \varphi\} \\
        \varphi_4 &:= \{(\overline{x_i} \lor \overline{x_j}) \mid i < j \ \text{and} \ (\overline{x_i} \lor \overline{x_j}) \in \varphi\}.
    \end{align*}
    It suffices to show that, for each $i \in [4]$, there is a 2-CNF formula $\varphi'_i$ such that $\varphi_i \bvachain \varphi'_i$ and $|\varphi'_i| \le \left(\frac{1}{4} + o(1)\right) \frac{n^2}{\lg n}$. By \Cref{thm:idealized-bva-monotone}, this is true for $i = 1$, and the case $i=4$ is the same after negating the variables. A nearly identical construction to that of \Cref{thm:idealized-bva-monotone} (with $p'(u,v) = +$ for each $v \in V(G)$ rather than $p'(u,v) = -$) works for $i = 2$, and the case $i=3$ is the same after negating the variables.
\end{proof}

\begin{lemma} \label{lem-count-2cnf}
    The number of 2-CNF formulas on $n$ variables is $4^{n^2-n}$.
\end{lemma}
\begin{proof}
    With $n$ variables, there are $2n$ literals. A clause consists of a pair of literals, the number of which is $\binom{2n}{2} = 2n^2-n$. But, we do not allow formulas to contain tautologous clauses (i.e., those of the form $(x_i \lor \overline{x_i})$), the number of which is $n$. Thus, there are $2n^2 - 2n$ clauses to choose from when constructing a 2-CNF formula, and any subset of these constitutes a 2-CNF formula, so the number of 2-CNF formulas on $n$ variables is $2^{2n^2-2n} = 4^{n^2-n}$.
\end{proof}

\begin{proposition} \label{prop-bva-no-preprocessing-lb}
    There is a 2-CNF formula $\varphi$ on $n$ variables such that, for every $\varphi'$ with $\varphi \bvachain \varphi'$, we have $|\varphi'| \ge \left(1 - o(1)\right) \frac{n^2}{\lg n}$.
\end{proposition}
\begin{proof}
    Let $g(m)$ be the number of 2-CNF formulas with at most $m$ clauses, and let $h(n)$ be the number of 2-CNF formulas on $n$ variables. By \Cref{lem-count-formulas,lem-count-2cnf}, $\lg g(m) \le (1+o(1)) m \cdot \lg m$ and $\lg h(n) \ge (2-o(1)) n^2$.
    
    If $\varphi \bvachain \varphi'$, then performing variable elimination (in the Davis--Putnam sense) on the auxiliary variables in $\varphi'$ yields $\varphi$. In particular, $\varphi'$ has enough information to uniquely reconstruct $\varphi$. Thus, if for every 2-CNF formula $\varphi$ on $n$ variables, there is a 2-CNF formula $\varphi'$ with at most $m$ clauses such that $\varphi \bvachain \varphi'$, then $g(m) \ge h(n)$, and consequently $\lg g(m) \geq \lg h(n)$.
    If $m \le (r+o(1))n^2/\lg n$ for some constant $r$, then $\lg m \le (2+o(1))\lg n$, and thus
    \[
        (2 - o(1)) n^2 \le \lg h(n) \le \lg g(m) \le (1+o(1)) m \cdot \lg m \le (2r+o(1)) n^2,
    \]
    from where $r \ge 1$; that is, $m \ge (1 - o(1)) n^2/\lg n$.
\end{proof}

\section{Proofs from \Cref{sec-amo}}

\lemcontract*
\begin{proof}
    Write $\mathcal{R} = (B,A,E,p)$. Suppose first that $y$ has in-degree 1. Let $(y',y)$ be the unique incoming edge to $y$. Let $\mathcal{R}' = (B,A \setminus \{y\},E',p')$ be defined by
    \begin{align*}
        E' &= (E \setminus \{e \in E \mid y \in e\}) \cup \{(y',z) \mid (y,z) \in E\} \\
        p'(u,v) &= \begin{cases}
            p(y,v) &\text{if } u = y' \ \text{and} \ v \in B \\
            p(u,v) & \text{otherwise}.
        \end{cases}
    \end{align*}
    Then, $\mathcal{R}'$ is an SPRN realizing $G$ and has at most $m-1$ edges. On the other hand, if $y$ has out-degree 1, then let $(y,y')$ be the unique outgoing edge from $y$. Let $\mathcal{R}' = (B,A \setminus \{y\},E',p')$ be defined by
    \begin{align*}
        E' &= (E \setminus \{e \in E \mid y \in e\}) \cup \{(z,y') \mid (z,y) \in E\} \\
        p'(u,v) &= \begin{cases}
            p(y,y') &\text{if } v = y' \in B \ \text{and} \ (u,y) \in E \\
            p(u,v) & \text{otherwise}.
        \end{cases}
    \end{align*}
    Then, $\mathcal{R}'$ is an SPRN realizing $G$ and has at most $m-1$ edges.
\end{proof}

\lemdegone*
\begin{proof}
    Let $y$ be the unique vertex adjacent to $x$ in $\mathcal{R}$. Then, $y$ must be an auxiliary vertex, or else there would be no valid walk between $x$ and $V(K_n) \setminus \{x,y\}$. Without loss of generality, suppose $(x,y)$ is the edge connecting these two vertices in $\mathcal{R}$. We claim that there is no base variable $x' \in V(K_n) \setminus \{x\}$ with a valid walk from $x'$ to $y$. Indeed, there is a valid walk $(x,y,\dots,x')$, so if there were a valid walk $(x',\dots,y)$, we could construct the valid walk $(x',\dots,y,\dots,x')$, contradicting our assumption that $\mathcal{R}$ is an SPRN realizing $K_n$. Since $\mathcal{R}$ is strict, it follows that $y$ has in-degree exactly 1. By \Cref{lem-contract}, there is an SPRN realizing $K_n$ with at most $m-1$ edges.
\end{proof}

\subsection{How BVA Reencodes \textsf{AtMostOne} in Practice} \label{sec-amo-in-practice}

Here, we prove that, using the heuristics for performing BVA steps described in \cite{mantheyAutomatedReencodingBoolean2012}, BVA always reencodes $\textsf{AtMostOne}(x_1,\dots,x_n)$ into a formula with $3n-6$ clauses for $n \ge 3$.

To prove this, we must specify how a BVA step is performed in practice. Following the graph-theoretic lens we have been using, we describe the algorithm as it operates on the diagram corresponding to a simple 2-CNF formula. Given a diagram $G = (V,E)$ and $v \in V$, let $N_G(v) = \{u \mid (v,u) \in E\} \cup \{\overline{u} \mid \{v,u\} \in E\}$. Then, a set of clauses of the form $\mathcal{C} \bowtie \mathcal{D}$ corresponds in $G$ to a complete bipartite subgraph $(L,R)$, where $R \subseteq N(v)$ for all $v \in L$; we call this a \emph{biclique}.

\Cref{alg:bva-step} is the algorithm for choosing a BVA step adapted from the original BVA paper~\cite{mantheyAutomatedReencodingBoolean2012}. The algorithm, as we have presented it, is nondeterministic since there may be multiple choices of $v$ in \ref{choose-v} and multiple iteration orders in \ref{loop-w}. Our result will apply to all possible executions of the algorithm, which means that it also applies to SBVA~\cite{haberlandtEffectiveAuxiliaryVariables2023}, which adds tie-breaking heuristics to the original BVA algorithm but still adheres to the specification of \Cref{alg:bva-step}.

\begin{algorithm}[h!]
  \caption{\textsc{Heuristic-BVA-Step}$(G)$}
  \label{alg:bva-step}
  \DontPrintSemicolon
  \KwIn{A diagram $G = (V,E)$}
  \KwOut{A biclique $(L,R)$ found by one BVA step (or $\bot$ if no improvement)}
  
  choose $v \in V$ such that $|N_G(v)|$ is maximum \nllabel{choose-v}\;
  $L \gets \{v\}$ \tcp*{left side}
  $R \gets N_G(v)$ \tcp*{right side (current common neighborhood)}
  $Q \gets |L|\cdot|R| - |L| - |R|$ \tcp*{current biclique quality}
  
  \While{\textsf{true}}{
    $(v^\star, R^\star, Q^\star) \gets (\bot, R, Q)$ \;
    \ForEach{$w \in V \setminus L$}{ \nllabel{loop-w}
      $R_w \gets R \cap N_G(w)$ \tcp*{new right side if $w$ is added to $L$}
      $Q_w \gets (|L|+1)\cdot|R_w| - (|L|+1) - |R_w|$ \tcp*{quality of $(L \cup \{w\}, R_w)$}
      
      \If{$Q_w > Q^\star$}{
        $(v^\star, R^\star, Q^\star) \gets (w, R_w, Q_w)$ \;
      }
    }
    
    \If{$v^\star = \bot$}{
      \textbf{break} \tcp*{no vertex improves the quality}
    }
    
    \tcp{Accept the best candidate and update $(L,R)$}
    $(L, R, Q) \gets (L \cup \{v^\star\}, R^\star, Q^\star)$ \;
  }
  
  \If{$Q > 0$}{
    \Return $(L,R)$ \tcp*{biclique used for a BVA substitution}
  }
  \Else{
    \Return $\bot$ \tcp*{no profitable biclique found}
  }
\end{algorithm}

\begin{theorem} \label{thm-amo-practice}
    If the BVA algorithm, each step of which is chosen according to \Cref{alg:bva-step}, is applied to the formula $\textsf{AtMostOne}(x_1,\dots,x_n)$ for some $n \ge 3$, then the resulting encoding has $3n-6$ clauses.
\end{theorem}
\begin{proof}
    Note that $K_n$ is the diagram corresponding to $\textsf{AtMostOne}(x_1,\dots,x_n)$ for some $n \ge 3$, so we analyze how \Cref{alg:bva-step} operates on $K_n$. For $n \ge 3$, let $f(n)$ be the number of edges in the diagram resulting from repeatedly performing BVA steps according to \Cref{alg:bva-step} starting with $K_n$. The proof is by induction on $n$. The base cases are when $n \in \{3,4\}$, in which case \Cref{alg:bva-step} returns $\bot$. Thus, BVA does nothing in these cases, so $f(n) = \binom{n}{2} = 3n-6$.

    Now suppose that $n \ge 5$. Throughout \Cref{alg:bva-step}, we have $|R| = n - |L|$, so $Q = |L| \cdot (n-|L|) - |L| - (n-|L|) = n \cdot |L| - |L|^2 -n$. \Cref{alg:bva-step} increments $|L|$ so long as $Q$ strictly increases, which happens until $|L| = \lfloor n/2 \rfloor$ and $|R| = n-|L| = \lceil n/2 \rceil$. Without loss of generality, $L = \{x_1,\dots,x_{\lfloor n/2 \rfloor}\}$ and $R = \{\overline{x_{\lfloor n/2 \rfloor + 1}},\dots,\overline{x_n}\}$. Thus, the BVA step replaces the clauses
    \[
        \{(\overline{x_i} \lor \overline{x_j}) \mid i \in [\lfloor n/2 \rfloor], j \in [\lfloor n/2 \rfloor + 1, n]\}
    \]
    by the clauses
    \[
        \{(\overline{x_i} \lor y) \mid i \in [\lfloor n/2 \rfloor]\} \cup \{(\overline{y} \lor \overline{x_j}) \mid j \in [\lfloor n/2 \rfloor + 1, n]\}.
    \]
    The resulting formula $\varphi(x_1,\dots,x_n,y)$ can be written as $\varphi_1 \land \varphi_2$, where
    \begin{align*}
        \varphi_1 &:= \textsf{AtMostOne}(x_1,\dots,x_{\lfloor n/2 \rfloor}, \overline{y}) \\
        \varphi_2 &:= \textsf{AtMostOne}(y,x_{\lfloor n/2 \rfloor + 1}, \dots, x_n).
    \end{align*}
    Now, we claim that the next biclique found by \Cref{alg:bva-step} (if there is one) must be entirely within $\varphi_1$ or $\varphi_2$. Indeed, for \Cref{alg:bva-step} to not return $\bot$, we must have $|L| \ge 2$. If $|L| \ge 2$, then $x_i \in L$ for some $i \in [n]$. If $i \le [\lfloor n/2 \rfloor]$, then the biclique is entirely contained within $\varphi_1$; otherwise, it is entirely contained within $\varphi_2$. Thus, $f(n) = f(\lfloor n/2 \rfloor + 1) + f(\lceil n/2 \rceil + 1)$. The solution to this recurrence is $f(n) = 3n-6$, as desired.
\end{proof}

\section{Experimental Details}
\label{sec:experimemtal-details}

\subparagraph*{Hardware.} We ran all experiments on a MacBook Pro personal computer, with 16 GB of RAM, an Apple M5 chip, and running macOS Tahoe 26.2; all experiments were single-threaded.

\subparagraph*{Code.} Our code is available at \url{https://github.com/bsubercaseaux/BicliqueVA}

\subparagraph*{Instances.} The formulas for independent sets of random graphs are created as follows. First, a graph $G \sim G(n, p)$ is sampled. Then, for each vertex $v$, out of $n$ vertices, is represented by a variable $x_v$, and thus for every edge $\{u, v\}$ we add a clause $(\overline{x_u} \lor \overline{x_v})$. Then, a cardinality constraint $\sum_{v} x_v \geq k$ is added, for which we use the sequential counter as implemented in PySAT~\cite{imms-sat18}. Since the expected independence number of a $G(n, \tfrac{1}{2})$ is well-concentrated around $2\lg n$, we consider both a satisfiable regime, for which we set $k := 1+\lfloor 1.2 \lg n \rfloor$, as well as an unsatisfiable regime for which we set $k := \lfloor 30 \lg n \rfloor$ in order to keep the runtimes manageable (as $k$ approaches the critical threshold, the instances become harder to solve).
Since randomness is involved, all our results, both for clauses, variables, and runtimes, are averaged over $5$ independent trials.

The concrete numbers corresponding to the bar plots in~\Cref{sec:discussion}, which are for the satisfiable regime, are presented in~\Cref{tab:sat-regime}. On the other hand, \Cref{tab:unsatregime} presents data on the unsatisfiable regime.


\begin{table}[htbp]
\caption{Experimental results for the SAT regime, summarizing variables (both base and auxiliary), clauses, and computational time across different graph sizes. The best of each group is highlighted in green, except for the number of variables, in which we do not consider the original encoding.}\label{tab:sat-regime}
\centering
\scalebox{0.90}{
\begin{tabular}{clrrrrr}
\toprule
$n$ & Method & \# Vars & \# Clauses & Reenc. Time (ms) & Solve Time (ms) & Total (ms) \\
\midrule
\multirow{6}{*}{600}
    & Original & 600 & 89671 & --- & 364.60 & 364.60 \\
    & \textsf{BiVA} & \colorbox{green!20!white}{3713} & 44190 & \colorbox{green!20!white}{37.63} & 319.82 & \colorbox{green!20!white}{357.45} \\
    & \textsf{BVA} & 4847 & \colorbox{green!20!white}{32404} & 273.99 & 333.87 & 607.86 \\
    & \textsf{factor} & 4551 & 32583 & 134.29 & 467.65 & 601.94 \\
    & \textsf{BiVA}+\textsf{BVA} & 6463 & 36719 & 99.39 & \colorbox{green!20!white}{308.18} & 407.57 \\
    & \textsf{BiVA}+\textsf{factor} & 6459 & 36700 & 66.52 & 310.99 & 377.51 \\
\midrule
\multirow{6}{*}{1200}
    & Original & 1200 & 359732 & --- & 1641.02 & 1641.02 \\
    & \textsf{BiVA} & \colorbox{green!20!white}{12582} & 156200 & \colorbox{green!20!white}{132.19} & 799.70 & 931.89 \\
    & \textsf{BVA} & 15862 & \colorbox{green!20!white}{117768} & 2433.71 & \colorbox{green!20!white}{439.00} & 2872.71 \\
    & \textsf{factor} & 14731 & 118606 & 964.25 & 876.73 & 1840.98 \\
    & \textsf{BiVA}+\textsf{BVA} & 21870 & 128923 & 495.17 & 484.75 & 979.92 \\
    & \textsf{BiVA}+\textsf{factor} & 22028 & 128976 & 272.41 & 566.34 & \colorbox{green!20!white}{838.75} \\
\midrule
\multirow{6}{*}{1800}
    & Original & 1800 & 809413 & --- & 3684.32 & 3684.32 \\
    & \textsf{BiVA} & \colorbox{green!20!white}{18899} & 322931 & \colorbox{green!20!white}{306.63} & 827.77 & \colorbox{green!20!white}{1134.40} \\
    & \textsf{BVA} & 32074 & \colorbox{green!20!white}{251632} & 11811.92 & 362.26 & 12174.19 \\
    & \textsf{factor} & 29706 & 253436 & 3244.75 & \colorbox{green!20!white}{352.23} & 3596.98 \\
    & \textsf{BiVA}+\textsf{BVA} & 39703 & 260504 & 1743.41 & 727.06 & 2470.47 \\
    & \textsf{BiVA}+\textsf{factor} & 39900 & 260705 & 754.32 & 430.94 & 1185.27 \\
\midrule
\multirow{6}{*}{2400}
    & Original & 2400 & 1439527 & --- & 8617.10 & 8617.10 \\
    & \textsf{BiVA} & \colorbox{green!20!white}{25200} & 548671 & \colorbox{green!20!white}{535.17} & 2599.50 & 3134.67 \\
    & \textsf{BVA} & 53234 & \colorbox{green!20!white}{431841} & 30201.91 & 1796.68 & 31998.60 \\
    & \textsf{factor} & 48955 & 435083 & 7618.12 & \colorbox{green!20!white}{1384.92} & 9003.04 \\
    & \textsf{BiVA}+\textsf{BVA} & 61290 & 432518 & 5002.72 & 1901.01 & 6903.74 \\
    & \textsf{BiVA}+\textsf{factor} & 61463 & 432747 & 1546.71 & 1410.66 & \colorbox{green!20!white}{2957.37} \\
\midrule
\multirow{6}{*}{3000}
    & Original & 3000 & 2249189 & --- & 15724.21 & 15724.21 \\
    & \textsf{BiVA} & \colorbox{green!20!white}{54359} & 830189 & \colorbox{green!20!white}{867.49} & 2312.35 & 3179.84 \\
    & \textsf{BVA} & 78895 & \colorbox{green!20!white}{657162} & 63638.74 & 1526.34 & 65165.08 \\
    & \textsf{factor} & 72535 & 662432 & 15205.77 & 1908.56 & 17114.34 \\
    & \textsf{BiVA}+\textsf{BVA} & 101018 & 668154 & 8617.16 & 849.28 & 9466.45 \\
    & \textsf{BiVA}+\textsf{factor} & 103489 & 669987 & 2398.41 & \colorbox{green!20!white}{570.11} & \colorbox{green!20!white}{2968.52} \\
\bottomrule
\end{tabular}
}
\label{tab:experimental_results}
\end{table}

\begin{table}[htbp]
\caption{Experimental results for the UNSAT regime, summarizing variables (both base and auxiliary), clauses, and computational time across different graph sizes. The best of each group is highlighted in green, except for the number of variables, in which we do not consider the original encoding.}
\centering
\scalebox{0.90}{
\begin{tabular}{clrrrrr}
\toprule
$n$ & Method & \# Vars & \# Clauses & Reenc. Time (ms) & Solve Time (ms) & Total (ms) \\
\midrule
\multirow{6}{*}{500}
    & Original & 500 & 62255 & --- & 129.47 & 129.47 \\
    & \textsf{BiVA} & \colorbox{green!20!white}{3078} & 31958 & \colorbox{green!20!white}{26.33} & \colorbox{green!20!white}{73.86} & \colorbox{green!20!white}{100.20} \\
    & \textsf{BVA} & 3572 & \colorbox{green!20!white}{23103} & 157.26 & 94.63 & 251.88 \\
    & \textsf{factor} & 3379 & 23241 & 81.51 & 95.35 & 176.86 \\
    & \textsf{BiVA}+\textsf{BVA} & 4994 & 26825 & 65.48 & 74.49 & 139.96 \\
    & \textsf{BiVA}+\textsf{factor} & 4980 & 26834 & 43.59 & 74.80 & 118.39 \\
\midrule
\multirow{6}{*}{1000}
    & Original & 1000 & 249714 & --- & 2144.22 & 2144.22 \\
    & \textsf{BiVA} & \colorbox{green!20!white}{10408} & 113808 & \colorbox{green!20!white}{95.34} & 1670.44 & \colorbox{green!20!white}{1765.78} \\
    & \textsf{BVA} & 11576 & \colorbox{green!20!white}{83811} & 1329.48 & 1649.19 & 2978.67 \\
    & \textsf{factor} & 10808 & 84326 & 568.98 & \colorbox{green!20!white}{1592.09} & 2161.07 \\
    & \textsf{BiVA}+\textsf{BVA} & 16854 & 94393 & 305.18 & 1686.26 & 1991.44 \\
    & \textsf{BiVA}+\textsf{factor} & 17019 & 94443 & 184.62 & 1679.89 & 1864.51 \\
\midrule
\multirow{6}{*}{1500}
    & Original & 1500 & 562196 & --- & 9075.38 & 9075.38 \\
    & \textsf{BiVA} & \colorbox{green!20!white}{15747} & 232248 & \colorbox{green!20!white}{213.32} & 7317.21 & \colorbox{green!20!white}{7530.53} \\
    & \textsf{BVA} & 23370 & \colorbox{green!20!white}{178873} & 5751.98 & 7535.59 & 13287.56 \\
    & \textsf{factor} & 21608 & 180030 & 1944.58 & 7573.61 & 9518.19 \\
    & \textsf{BiVA}+\textsf{BVA} & 30273 & 189617 & 1003.89 & 7869.11 & 8873.00 \\
    & \textsf{BiVA}+\textsf{factor} & 30446 & 189677 & 479.09 & \colorbox{green!20!white}{7103.80} & 7582.88 \\
\midrule
\multirow{6}{*}{2000}
    & Original & 2000 & 999401 & --- & 69761.13 & 69761.13 \\
    & \textsf{BiVA} & \colorbox{green!20!white}{20982} & 392056 & \colorbox{green!20!white}{396.73} & 58099.24 & 58495.96 \\
    & \textsf{BVA} & 38593 & \colorbox{green!20!white}{306556} & 17114.92 & 63503.85 & 80618.77 \\
    & \textsf{factor} & 35646 & 308821 & 4477.50 & 59651.54 & 64129.04 \\
    & \textsf{BiVA}+\textsf{BVA} & 46496 & 313607 & 2721.44 & 60768.45 & 63489.89 \\
    & \textsf{BiVA}+\textsf{factor} & 46671 & 313805 & 1013.76 & \colorbox{green!20!white}{57190.54} & \colorbox{green!20!white}{58204.30} \\
\bottomrule
\end{tabular}
}
\label{tab:unsatregime}
\end{table}

\end{document}